%% file: orbifold_vis.tex
\documentclass[journal]{vgtc}                

\ifpdf
  \pdfoutput=1\relax                   
  \usepackage{graphicx}                
  \DeclareGraphicsExtensions{.pdf,.png,.jpg,.jpeg} 
\else
  \usepackage{graphicx}                
  \DeclareGraphicsExtensions{.eps}     
\fi%

\graphicspath{{figures/}{pictures/}{images/}{./}} 

\usepackage{microtype}                 
\PassOptionsToPackage{warn}{textcomp}  
\usepackage{textcomp}                  
\usepackage{mathptmx}                  
\usepackage{times}                     
\usepackage{cite}                      
\usepackage{tabu}                      
\usepackage{booktabs}                  
\usepackage{comment}                   
\usepackage[ruled]{algorithm2e}        

\usepackage{wrapfig}
\usepackage{todonotes}
\usepackage{mathptmx}
\usepackage{graphicx}
\usepackage{times}
\usepackage{amsmath,amscd,amssymb, amsthm}
\usepackage{times}
\usepackage[abs]{overpic}
\usepackage{cases}
\usepackage{float}
\usepackage{caption}
\usepackage{subfloat}
\usepackage{subfig}
\usepackage{tabularx}
\usepackage{arydshln}

\onlineid{1319}

\input{setting}

\vgtccategory{Research}
\vgtcpapertype{algorithm/technique}

\title{Interactive Design and Optics-Based Visualization of Arbitrary Non-Euclidean Kaleidoscopic Orbifolds}

\author{Jinta Zheng, Eugene Zhang, \textit{Senior Member, IEEE}, and Yue Zhang, \textit{Member, IEEE}}
\authorfooter{
\item
Jinta Zheng is with Oregon State University. E-mail: zhenjint@eecs.oregonstate.edu.
\item
Eugene Zhang is with Oregon State University. E-mail: zhange@eecs.oregonstate.edu.
\item
Yue Zhang is with Oregon State University. E-mail: zhangyue@oregonstate.edu.
}

\shortauthortitle{Zheng \MakeLowercase{\textit{et al.}}: Interactive Design and Optics-Based Visualization of Arbitrary Non-Euclidean Kaleidoscopic Orbifolds}

\abstract{Orbifolds are a modern mathematical concept that arises in the research of hyperbolic geometry with applications in computer graphics and visualization. In this paper, we make use of rooms with mirrors as the visual metaphor for orbifolds. Given any arbitrary two-dimensional kaleidoscopic orbifold, we provide an algorithm to construct a Euclidean, spherical, or hyperbolic polygon to match the orbifold. This polygon is then used to create a room for which the polygon serves as the floor and the ceiling. With our system that implements M\"obius transformations, the user can interactively edit the scene and see the reflections of the edited objects. To correctly visualize non-Euclidean orbifolds, we adapt the rendering algorithms to account for the geodesics in these spaces, which light rays follow. Our interactive orbifold design system allows the user to create arbitrary {\em two-dimensional kaleidoscopic} orbifolds. In addition, our mirror-based orbifold visualization approach has the potential of helping our users gain insight on the orbifold, including its orbifold notation as well as its universal cover, which can also be the spherical space and the hyperbolic space.
}


\CCScatlist{ 
 \CCScat{K.6.1}{Management of Computing and Information Systems}%
{Project and People Management}{Life Cycle};
 \CCScat{K.7.m}{The Computing Profession}{Miscellaneous}{Ethics}
}

\teaser{
\centering
 $\begin{array}{@{\hspace{0.0in}}c@{\hspace{0.05in}}c@{\hspace{0.05in}}c}
 \includegraphics[width=2.2in]{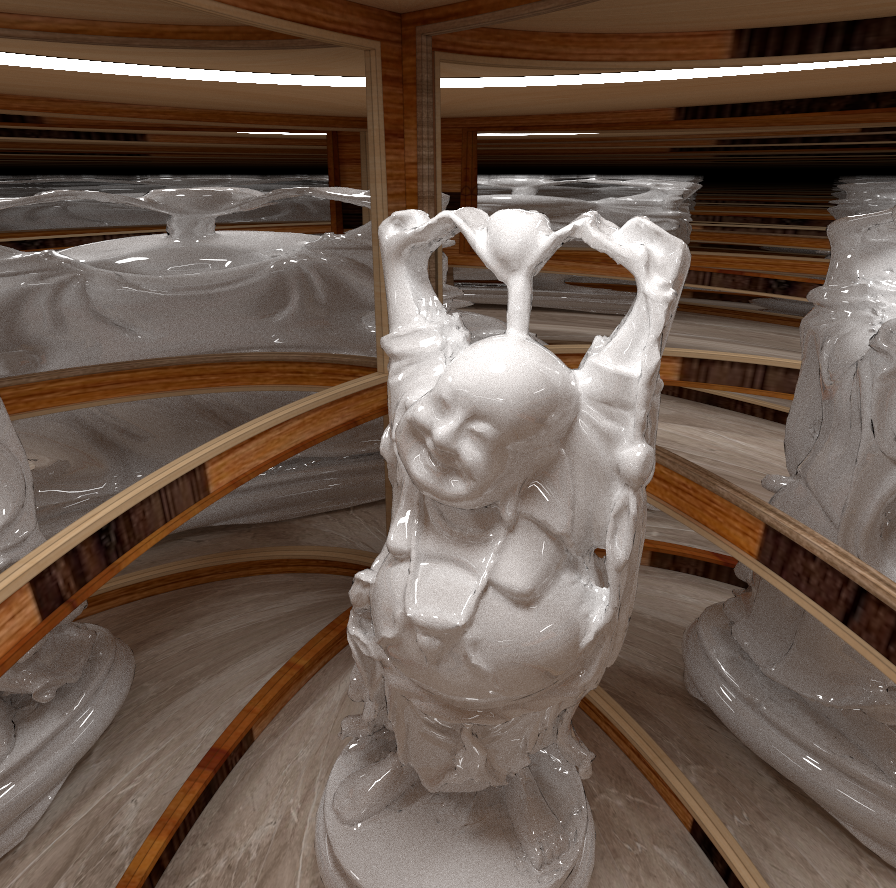}
 &\includegraphics[width=2.19in]{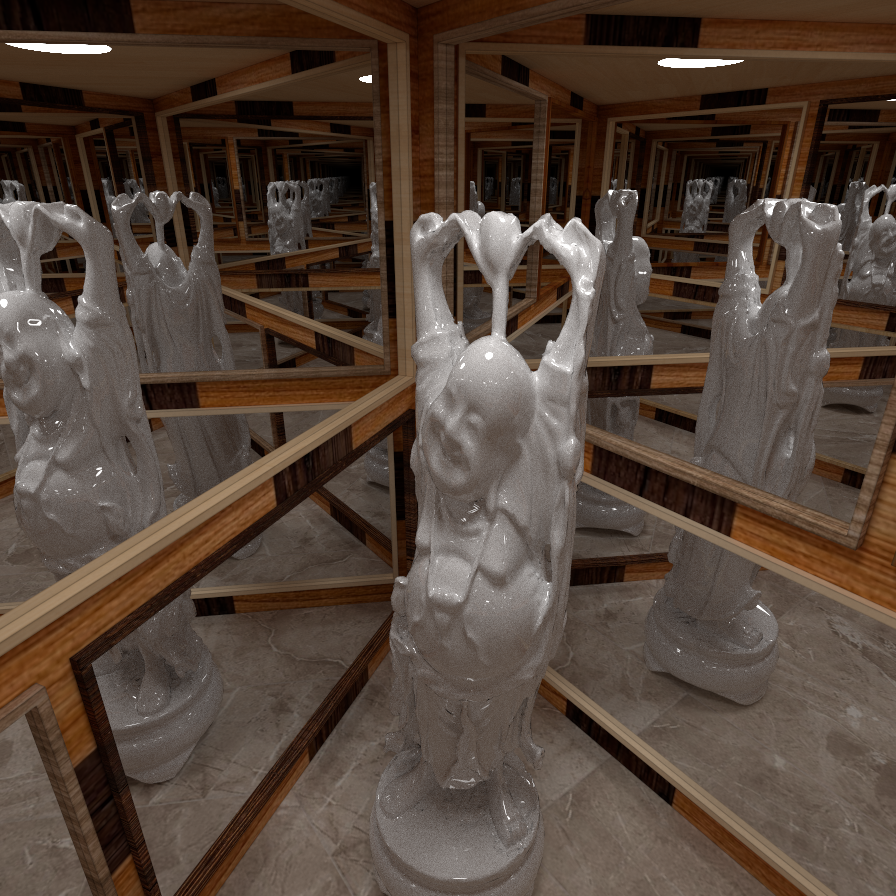}
 &\includegraphics[width=2.2in]{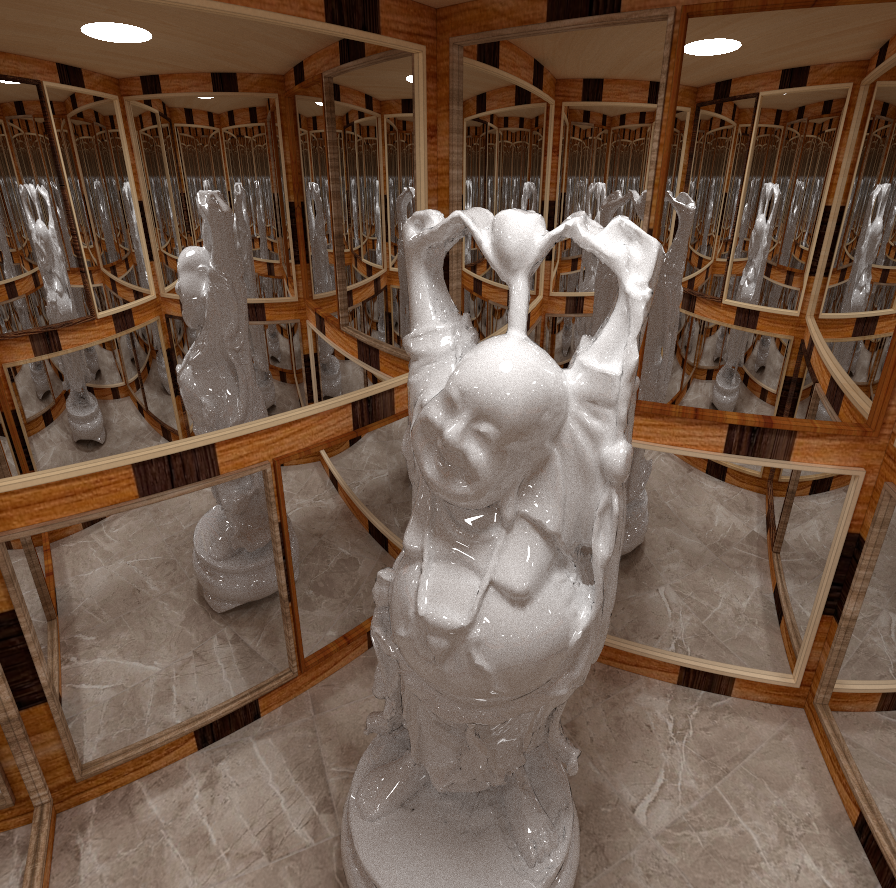} \\
 (a) & (b) & (c)
\end{array}$

\caption[]{Our mirror-metaphor visualization of orbifolds: (a) a spherical orbifold, (b) a Euclidean orbifold, and (c) a hyperbolic orbifold. Notice the unusual deformation in the reflections of Buddha and the beams in (a) and (c). All of the orbifolds in this figure were created using our orbifold construction algorithm (Section~\ref{sec:construction}). }
\label{fig:teaser}
}

\renewcommand{\manuscriptnotetxt}{}

\vgtcinsertpkg

\begin{document}



\firstsection{Introduction} 
\label{sec:intro}

\maketitle

\begin{figure*}[t]
	\centering%
	$\begin{array}{@{\hspace{0.0in}}c@{\hspace{0.1in}}c@{\hspace{0.1in}}c}
	\includegraphics[width=2.1in]{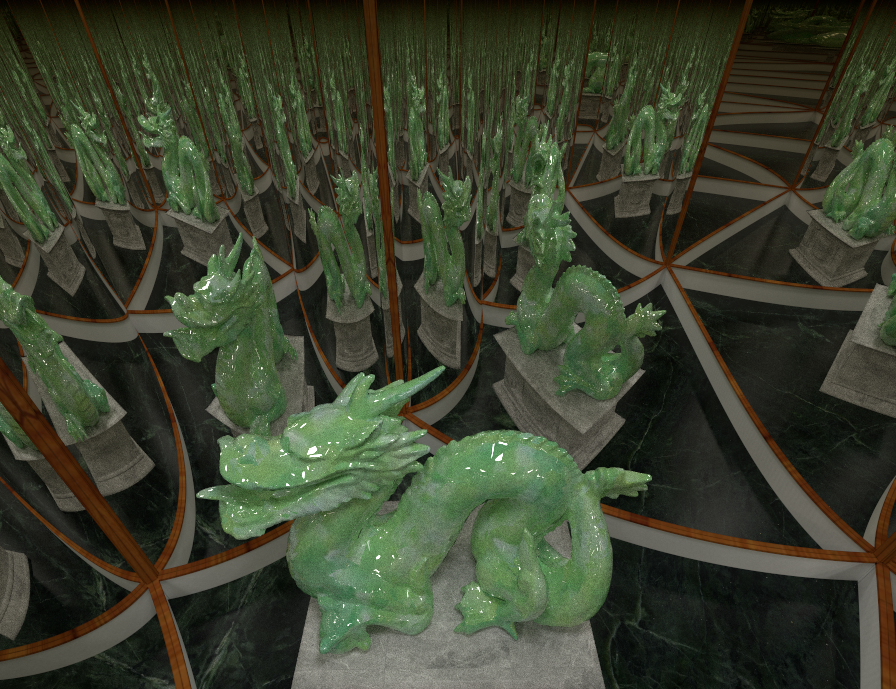}
	&\includegraphics[width=2.1in]{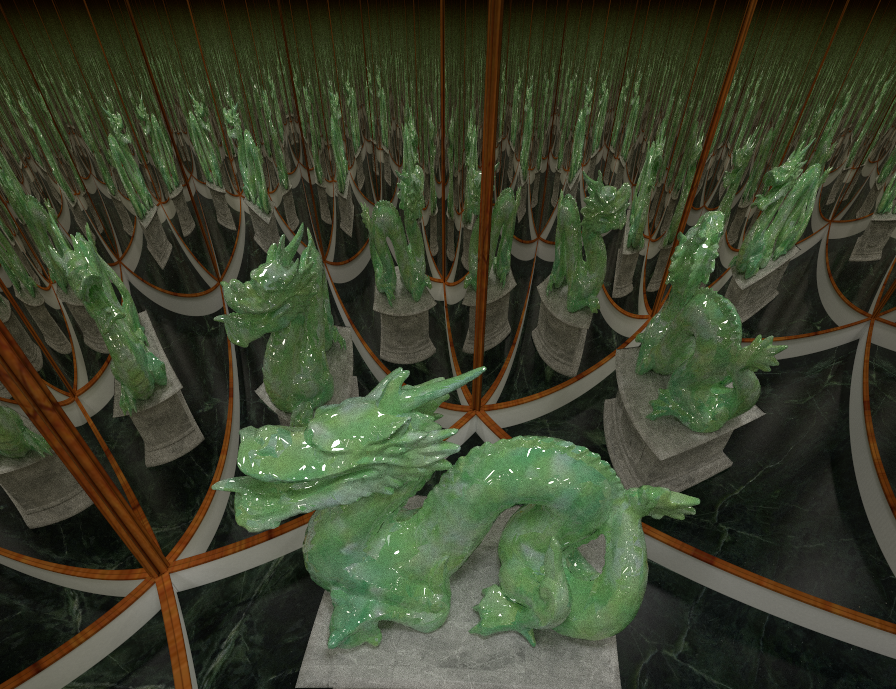}
	&\includegraphics[width=2.16in]{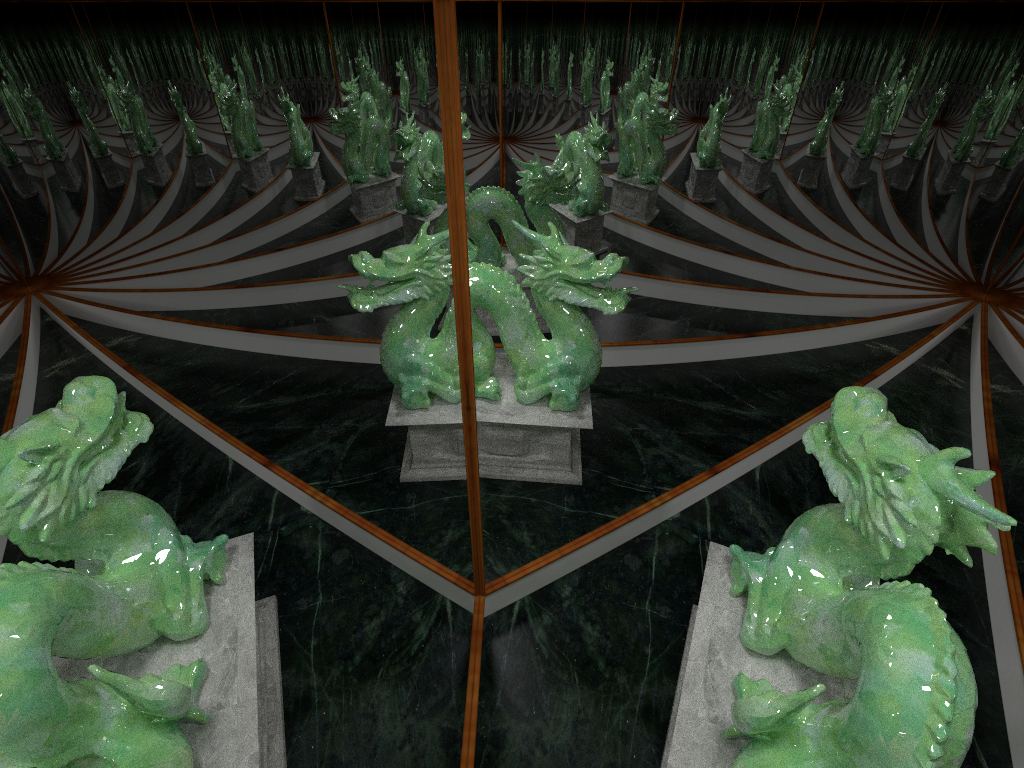}
\\
    (a) & (b) &(c) \\
	\end{array}$
	\caption[]{In (a-b), we show the comparison of the rendering of a hyperbolic orbifold with straight rays as in the Euclidean space (a) and the geodesics in the hyperbolic space (b).  In (c), we render a triangular non-orbifold with corner angles \textbf{$\frac{2\pi}{3}$-$\frac{\pi}{12}$-$\frac{\pi}{12}$}. Notice the dragon now has two heads.  }
	\label{fig:str_curved_rays}
\end{figure*}

\begin{figure}[t]
	\centering%
	$\begin{array}{@{\hspace{0.0in}}c}
	\includegraphics[width=3.5in]{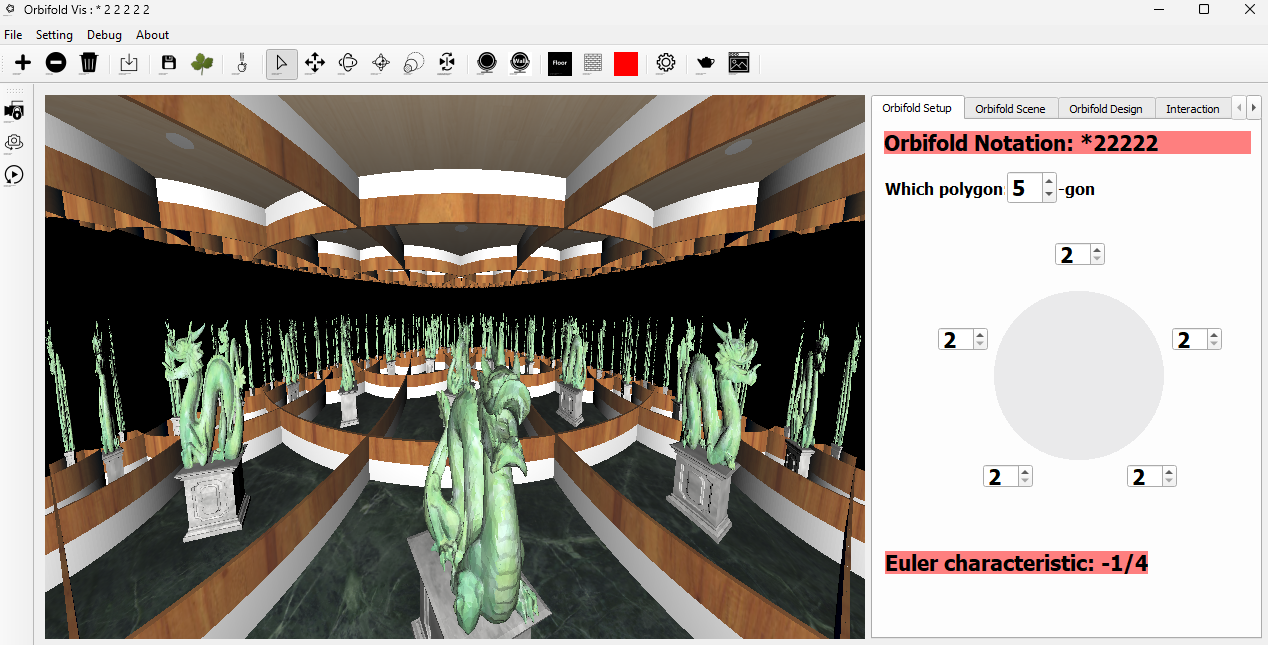} \\
	\end{array}$
	\caption[]{The interface of our design system.}
	\label{fig:interactive_system}
\end{figure}

Orbifolds are a modern mathematical concept originated from the study of low-dimensional topology~\cite{thurston1997three}. This concept has been used to study the geometric structures of hyperbolic spaces and to prove in $2003$ the famous Poincar\'e conjecture for the three-dimensional case, which is the last case and thus the hardest case to be addressed~\cite{Shea:07}. In addition, orbifolds have been used to describe spatial symmetries in String Theory~\cite{Giaccari_2023}. Orbifold theory has found applications in tensor field topology~\cite{Palacios:16,Roy:19,Qu:21}, remeshing~\cite{Kaelberer:07,Bommes:09,Nieser:12}, non-photorealistic rendering~\cite{Hertzmann:00,Zhang:07,Palacios:07}, and texture synthesis~\cite{Aigerman:2015,aigerman2016hyperbolic,Aigerman:2017,Palacios:17}.

The concept of orbifolds can be difficult to digest, as its definition involves various mathematical notions from topology and abstract algebra, such as {\em Hausdorff spaces}, {\em groups and group actions}, {\em charts}, and {\em gluing maps}. The problem is further compounded by the fact that most of the orbifolds are non-Euclidean, thus making the understanding and visualization of orbifolds more difficult. In this paper, we focus on two-dimensional {\em kaleidoscopic} orbifolds that are generated by reflections in the Euclidean plane, the sphere, and the hyperbolic space. We have built an interactive system with which our users can create any arbitrary two-dimensional kaleidoscopic orbifolds and interact with them to gain insight into the orbifolds.

Most existing work on visualizing two-dimensional orbifolds make use of texture patterns that tile the plane seamlessly.
Such an approach is commonly used to create illustrations. We observe the recent trend of engaging learning through 3D graphics and animations and employ an approach inspired by the kaleidoscopes. By placing some simple objects inside the kaleidoscope, fascinating images appear when we look through the viewing hole. To be more engaging, our system promotes a visual experience of being inside the kaleidoscope entirely. We utilize a room with mirrors as a visual metaphor for a two-dimensional orbifold.

Interestingly, the reflectional symmetries in a kaleidoscope correspond to the behaviors of a particular Euclidean orbifold. As shown in Figure~\ref{fig:teaser}, our system produces orbifolds (configuration of the ceiling and the floor) and the symmetry that each orbifold induces. In addition, through the bending of the mirror frames and the unfamiliar deformations of Buddha in (a) and (c), the notions of spherical geometry and hyperbolic geometry are visually delivered, respectively.

While there are only a handful of Euclidean orbifolds, there are infinitely many spherical and hyperbolic orbifolds. In fact, any polygon whose corner angles can each be expressed as $\frac{\pi}{k}$ ($k \in \mathbb{N}^+$) is an orbifold. To the best of our knowledge, there is no algorithm published that allows the realization of arbitrary such polygons when their natural spaces are hyperbolic. Most available tools focus on regular polygons. For arbitrary polygons that represent an orbifold, the lengths of the edges are challenging to determine.

In this paper, we address this difficulty by providing an algorithm that can interactively realize any orbifold, whether spherical, hyperbolic, or Euclidean (Section~\ref{sec:construction}). As part of our algorithm, we provide a complete enumeration of two-dimensional kaleidoscopic orbifolds based on the cardinality of the underlying polygon and the type of the {\em universal cover} (Section~\ref{sec:enumeration}). With this ability, any two-dimensional orbifold can be converted to a room, whose ceiling and floor have the configuration of the polygon. Our system further allows interactive scene editing, with the room and its virtual copies being visible at the same time (Figure~\ref{fig:interactive_system}). Furthermore, the creation of the reflected copies in the underlying space ({\em universal cover}) that is either the sphere or the hyperbolic space can reinforce the perception of geometric deformations of these non-Euclidean spaces. To achieve this, we present a system to generate the universal cover of any polygonal orbifold and provide interactive updates to the virtual rooms in the universal cover through M\"obius transformations.

Light rays travel along the geodesics in the spherical and hyperbolic spaces. When rendering a non-Euclidean scene using Euclidean straight lines, incorrect appearances result as shown in Figure~\ref{fig:str_curved_rays} (a). We modify the rendering algorithms to account for the correct paths for rays (Figure~\ref{fig:str_curved_rays} (b)). In addition, by adjusting the attenuation of some or all the mirrors in the scene (Figure~\ref{fig:emphasis}), we can further emphasize the orbifold itself (a), intensify the emphasis of the universal cover (c), or highlight the {\em translational cover} of the orbifold (the room and an adjacent virtual room). With our design system and visualization, users can customize their orbifolds for their purposes.

\begin{figure*}[t]
	\centering%
	$\begin{array}{@{\hspace{0.0in}}c@{\hspace{0.1in}}c@{\hspace{0.1in}}c}
	\includegraphics[width=2.2in]{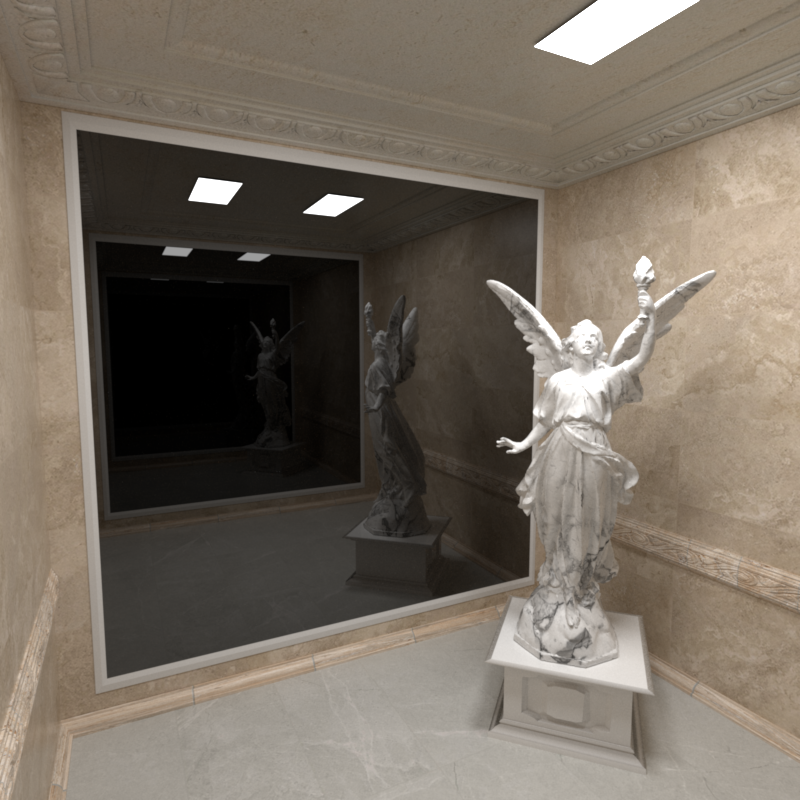}
&\includegraphics[width=2.2in]{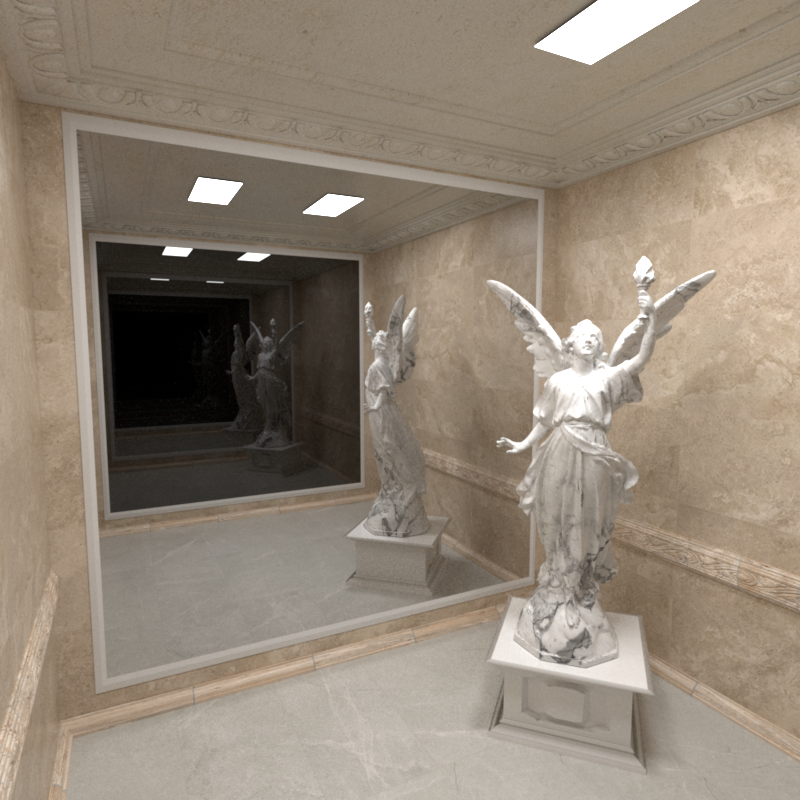}
	&\includegraphics[width=2.2in]{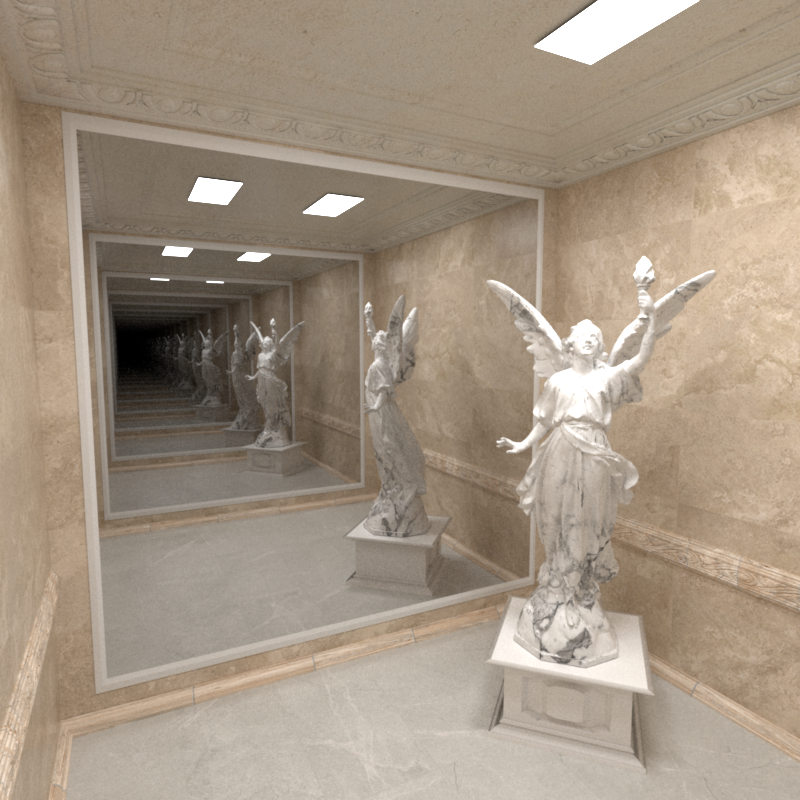} \\
	(a) & (b)  & (c)
	\end{array}$
	\caption[]{By varying the attenuations of the mirrors, our visualization can emphasize the orbifold (a), its translational cover (b), or its universal cover (c). }
	\label{fig:emphasis}
\end{figure*}

\section{Related Work}\label{sec:related_work}

Our work follows recent research in mathematics visualization, such as quaternions~\cite{Hanson:2005}, knots and links~\cite{vanWijk:05}, and branched covering spaces~\cite{Roy:18}.

Orbifolds are a modern mathematical concept~\cite{thurston1997three}. The notion of orbifolds has found applications in texture synthesis~\cite{Aigerman:2015,aigerman2016hyperbolic,Aigerman:2017}. Nieser and Polthier provide visualization of analytic functions over the complex plane~\cite{Nieser:10}, while Roy et al.~\cite{Roy:18} visualize the notion and properties of branched covering spaces with respect to $N$-way rotational symmetry ($N$-RoSy) fields~\cite{Palacios:07}. All of this work focuses on the visualization and processing of orbifolds on two-dimensional surfaces. Moreover, the orbifolds that they address only possess rotational symmetries, which are a subclass of reflectional symmetries. In this paper, we address orbifolds that are generated by reflections, which include not only reflectional symmetries but also rotational and translational symmetries.

Conway et al.~\cite{conway2008symmetries} explain the concepts and results related to planar orbifolds by using popular artwork containing textures with symmetries. Their approach focuses on Euclidean orbifolds. In our work, we provide a system to generate a room that matches {\em any} given two-dimensional kaleidoscopic orbifold, even when its universal cover is a non-Euclidean space. Furthermore, we make use of the mirror metaphor to leverage real-life experience with mirrors, which provides a complementary approach to the texture-based visualization of orbifolds.

Our mirror metaphor turns a two-dimensional orbifold into a three-dimensional room, which can also be considered as a three-dimensional orbifold that is the product of the two-dimensional orbifold (floor and ceiling) with a line segment (the height of the room). There has been some past research on visualizing three-dimensional orbifolds~\cite{Berger:2015,Novello:2020}, with a focus on the three-dimensional sphere $\mathbb{S}^3$ and three-dimensional hyperbolic space $\mathbb{H}^3$. In these spaces, the geodesics are either a circular arc or a hyperbola. In contrast, the geodesics in our product spaces are spirals, which makes ray-triangle intersection different from those in $\mathbb{S}^3$ and $\mathbb{H}^3$. Moreover, past research often focuses on using some famous orbifolds in $\mathbb{S}^3$ and $\mathbb{H}^3$ such as the Poincar\'e sphere and the mirror dodecahedron. In our paper, we allow the visualization of any arbitrary two-dimensional orbifolds.

\section{Background on Orbifolds}

In this section, we review necessary mathematical background on orbifolds used in this paper which include the concepts of {\em groups} and {\em group actions}~\cite{Humphreys:96}, {\em orbifolds}~\cite{Milley:98,Cooper:00}, and {\em non-Euclidean spaces}~\cite{coxeter1998non}. For a rigorous definition of these concepts, we refer our readers to the aforementioned references.

\begin{figure}[t]
	\centering%
	$\begin{array}{@{\hspace{0.0in}}c}
	\includegraphics[width=2.5in]{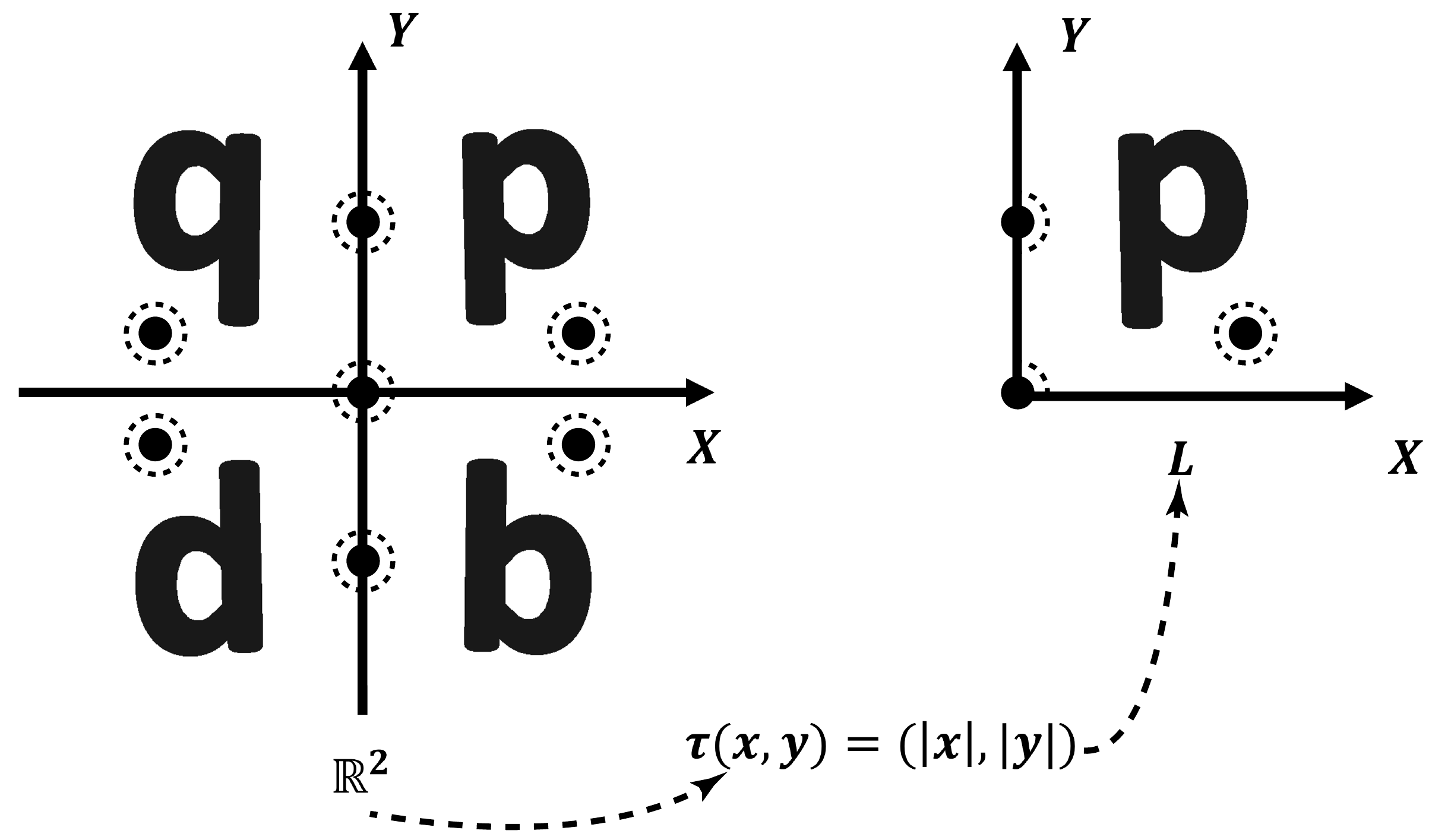}
	\end{array}$
	\caption[]{An orbifold (right: the first quadrant in the real plane) and its universal cover (left: all four quadrants) are related by a covering map $\tau(x, y)=(|x|, |y|)$.  The neighborhood of the origin in the orbifold (right) is a quarter disk. The map $\tau$ introduces a symmetry group consisting a horizontal reflection (the letter $p$ to $q$), a vertical reflection (the letter $p$ to $b$), and a rotation by $\pi$ (the letter $p$ to $d$).  The symmetry group is the Dihedral group $\mathbb{D}_2$.}
	\label{fig:bcs}
\end{figure}

An orbifold $O$ is a topological space $X$ paired with a discrete symmetry group $G$ such that $X$ locally resembles a Euclidean disk under the action of $G$. To better illustrate this, consider the space $L=\{(x, y) | x, y \ge 0 \}$ (Figure~\ref{fig:bcs} (right)). For each point in the first quadrant, we can find a small enough disk-shaped neighborhood. However, for a point on the positive $Y$-axis, there is a neighborhood of the shape of a half-disk which corresponds to a full disk in the Euclidean plane (Figure~\ref{fig:bcs} (left)) under the reflection across the $Y$-axis. In other words, the union of the half disk in the first quadrant and its mirror reflection form a full disk. Finally, the origin has a quarter-disk-shaped neighborhood (Figure~\ref{fig:bcs} (right)) which corresponds to a full disk in the Euclidean plane (Figure~\ref{fig:bcs} (left)) when being combined with its reflection across the $X$-axis (in the fourth quadrant), the reflection across the $Y$-axis (in the second quadrant), and the rotation by $\pi$ around the origin (in the third quadrant). Thus, $L$ is an orbifold.

Globally, we can see that $L$ is the range of the following function $\tau(x, y)=(|x|, |y|)$, which introduces a map from $\mathbb{R}^2$ to $L$ with the symmetry illustrated as follows. The letter $p$ (Figure~\ref{fig:bcs} (right)) corresponds to the letter $q$ in the second quadrant (Figure~\ref{fig:bcs} (left)) through the reflection across the $Y$-axis and the letter $b$ in the fourth quadrant through the reflection across the $X$-axis. In addition, it corresponds to the letter $d$ in the third quadrant through a rotation of $\pi$ around the origin, which is a composition of the two aforementioned reflections. Thus, the symmetry induced by the map $\tau$ leads to a symmetry group of four elements: the identity, two reflections, and one rotation. The group is the {\em Dihedral group} of order $2$, i.e. $\mathbb{D}_2$, which, when acted on $\mathbb{R}^2$, leads to the orbifold $L$. It has a {\em corner point} at the origin and two mirror lines (the positive $X$-axis and the positive $Y$-axis).

\begin{figure}[t]
	\centering%
	$\begin{array}{@{\hspace{0.0in}}c}
	\includegraphics[width=3.2in]{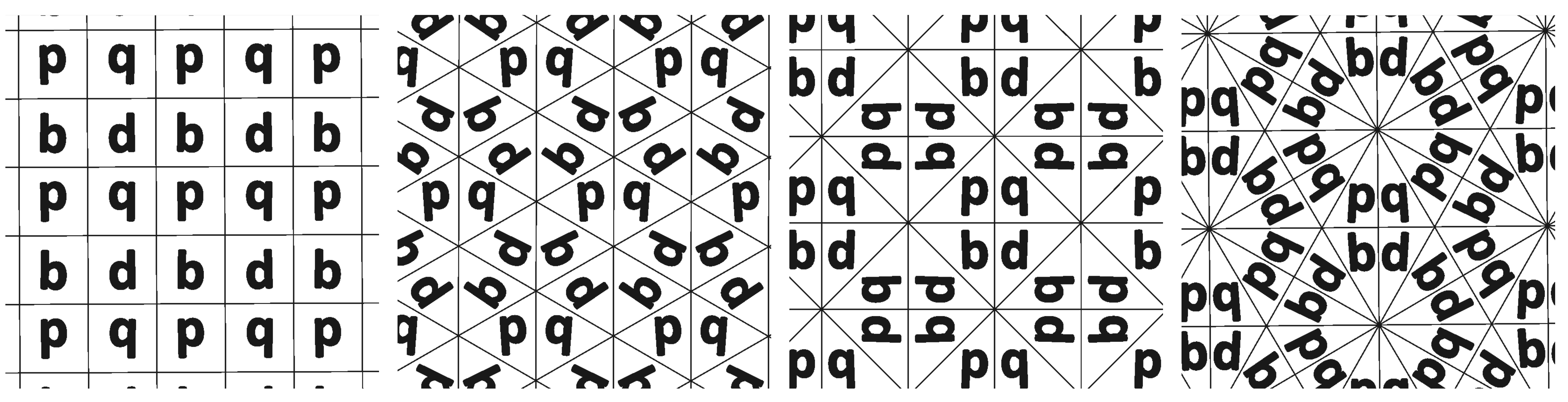} \\
	\put( -110,0){(a)  \textbf{$*2222$}}
	\put( -45,0){(b)  \textbf{$*333$}}
	\put( 13,0){(c)  \textbf{$*244$}}
	\put( 75,0){(d)  \textbf{$*236$}}
	\end{array}$
	\caption[]{The four Euclidean orbifolds.}
	\label{fig:2d_Euclidean}
\end{figure}

In general, a two-dimensional {\em kaleidoscopic} orbifold $O$ is a polygon with a symmetry group induced by reflections across all of its edges. Without causing ambiguity, we also refer to the polygon as $O$. Each edge of the polygon is thus a mirror line, and every vertex of the polygon is a corner point corresponding to the symmetry of $\mathbb{D}_k$, the dihedral group of order $k$. Note that $\mathbb{D}_k$ consists of $k$ rotations (including the identity) and $k$ reflections. In Figure~\ref{fig:2d_Euclidean} we show four such orbifolds, whose polygons have a configurations of a square (a), a $60^\circ-60^\circ-60^\circ$ triangle (b), a $90^\circ-45^\circ-45^\circ$ triangle (c), and a $90^\circ-60^\circ-30^\circ$ triangle (d). These orbifolds are given the {\em orbifold notations} (a) \textbf{$*2222$}, (b) \textbf{$*333$}, (c) \textbf{$*244$}, and (d) \textbf{$*236$}, respectively. A generic kaleidoscopic orbifold corresponding to an $N$-gon $O$ is given the notation \textbf{$*k_1 ... k_N$} where the \textbf{$*$} indicates the existence of the mirror and \textbf{$k_i$} implies that the angle of the polygon at the $i$-th corner is $\frac{\pi}{k_i}$.

An orbifold (the polygon) and all of its virtual copies through its symmetry group can seamlessly tile a space, which is its {\em universal cover}. The aforementioned orbifolds are kaleidoscopic orbifolds whose universal cover is the Euclidean plane, thus {\em Euclidean orbifolds}. Each Euclidean orbifold has a {\em translational cover}, which, along with its translational copies, form the universal cover. The translational cover of \textbf{$*2222$} consists of the orbifold, two of its reflections, and one rotation by $\pi$ (Figure~\ref{fig:2d_Euclidean} (a): any $2\times 2$ subgrid with the letters $q$, $p$, $d$, and $b$). The translations needed to generate the universal cover is the Gaussian integer grid $\mathbb{Z}[\textit{i}]$~\cite{GREAVES:12}. The translational covers of the other Euclidean orbifolds respectively consist of six copies arranged in a hexagon (Figure~\ref{fig:2d_Euclidean} (b): \textbf{$*333$}), eight copies arranged in an octagon (Figure~\ref{fig:2d_Euclidean}(c): \textbf{$*244$}), and twelve copies arranged in a dodecagon (Figure~\ref{fig:2d_Euclidean}(d): \textbf{$*236$}). The set of translations for \textbf{$*244$} is also $\mathbb{Z}[\textit{i}]$. On the other hand, the set of translations for \textbf{$*333$} and \textbf{$*236$} is the Eisenstein integer grid $\mathbb{Z}[\textit{$\omega$}]$~\cite{GREAVES:12} where $\omega = \frac{-1+\sqrt 3 i}{2}$.

While it may seem that these are the only kaleidoscopic orbifolds and that all kaleidoscopic orbifolds must be triangular or rectangular, there are many more. In fact, given an arbitrary polygon with at least three sides and whose corner angles divide $\pi$ individually, there is an orbifold that corresponds to the polygon. Figure~\ref{fig:non_Euclidean_room} shows a room with three, four, and five mirrors, respectively. However, these orbifolds cannot tile the Euclidean plane as their universal covers are either the unit sphere (spherical orbifolds) or the hyperbolic plane (hyperbolic orbifolds). The hyperbolic plane can be modeled as the upper sheet of the double-sheet hyperboloid $z^2 - x^2 - y^2=1$ (Figure~\ref{fig:hyper_hyperbolid}). Like the Euclidean orbifolds, both spherical orbifolds and hyperbolic orbifolds are polygons whose edges follow the geodesics in their universal cover. The geodesics in the unit sphere are the great circles, and the geodesic passing through two mutually distinct points $p$ and $q$ in the hyperbolic space is the intersection of the plane containing $p$, $q$ and the vertex of the lower-sheet of the hyperboloid (Figure~\ref{fig:hyper_hyperbolid}: the curve passing through $p$ and $q$). In fact, given an orbifold $O=$\textbf{$*k_1 ... k_N$} where $N$ is the number of walls and $k_i >1$ ($1\le i\le N$), its universal cover is decided by the {\em Euler characteristic} of the orbifold as follows:

\begin{figure}[t]
	\centering%
	$\begin{array}{@{\hspace{0.0in}}c}
	\includegraphics[width=2.1in]{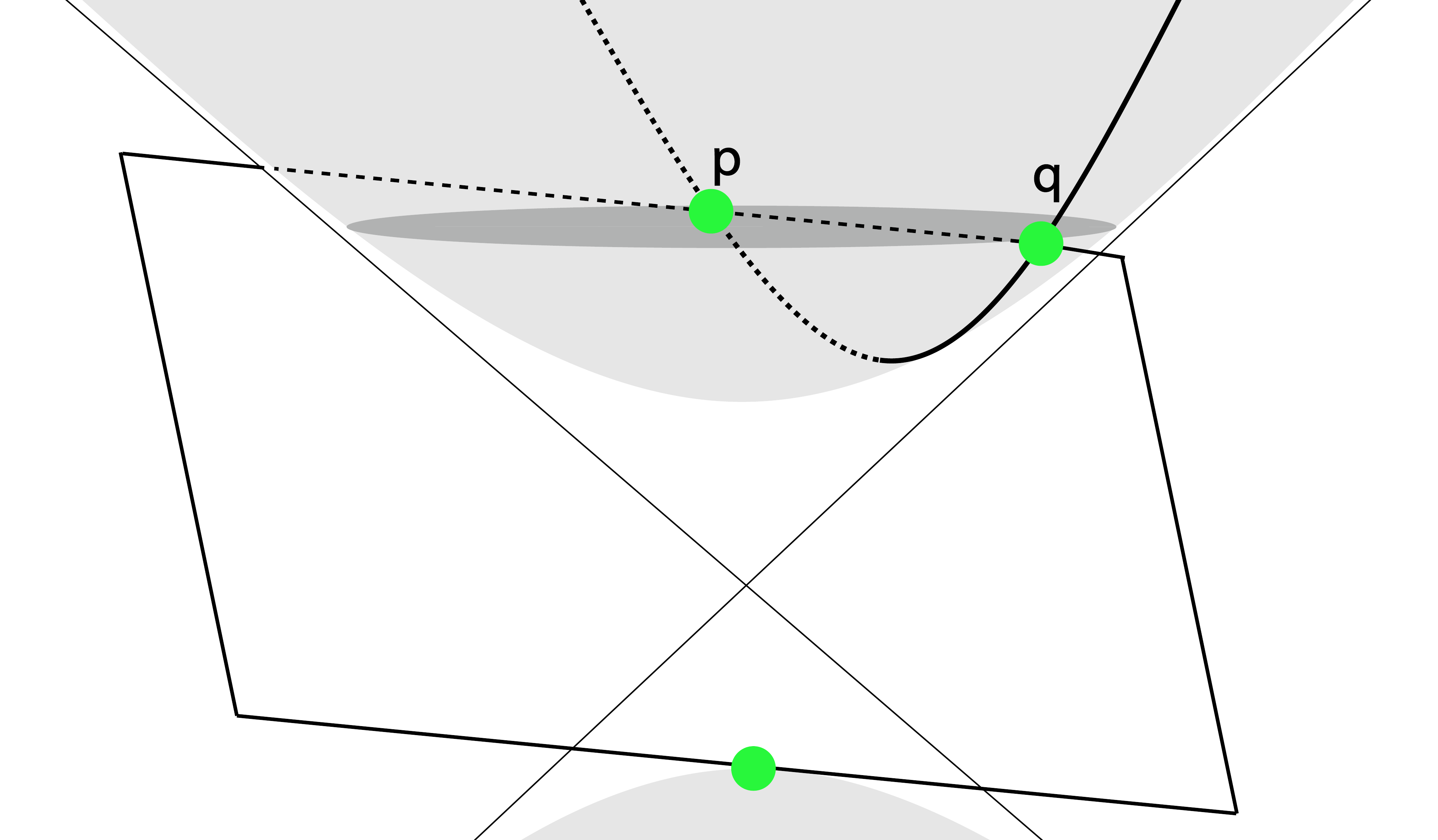} \\
	\end{array}$
	\caption[]{The hyperbolic space can be modeled as the upper-sheet of a double-sheet hyperboloid. The geodesic passing through $p$ and $q$ in the space is the intersection of the hyperboloid with the plane that passes through $p$ and $q$ as well as the vertex of the lower sheet. }
	\label{fig:hyper_hyperbolid}
\end{figure}

\begin{equation}\label{eq:euler}
  \chi(O) = \sum_{i=1}^{N} \frac{1}{2k_i} -\frac{N}{2}+1.
\end{equation}
An orbifold $O$ is spherical, Euclidean, or hyperbolic when $\chi(O)>0$, $\chi(O)=0$, $\chi(O)<0$, respectively.

In the next section, we describe our orbifold design system starting with an enumeration of all two-dimensional kaleidoscopic orbifolds.

\section{Kaleidoscopic Orbifold Enumeration}
\label{sec:enumeration}

While there has been a complete enumeration of spherical and Euclidean orbifolds, to our best knowledge such an enumeration is not explicitly given for hyperbolic orbifolds. In addition, the enumeration for spherical and Euclidean orbifolds is in the form of an exhaustive list. Our orbifold design system is based on the number of walls (the cardinality of the underlying polygon) in the orbifold. Thus, we strive for an {\em explicit} enumeration for all two-dimensional kaleidoscopic orbifolds based on the combination of the polygon cardinality and the universal cover.

There are three types of spherical orbifolds: (1) one mirror, (2) two mirrors, and (3) three mirrors. The only one mirror spherical orbifold is \textbf{$*$}, which corresponds to a room that is half of the sphere with its boundary being the mirror. There are no corners. In this case, one can consider the room as a monogon. In the second case, the room has two mirrors that intersect at $\frac{\pi}{k}$ at both ends where $k>1$. These are diangular orbifolds \textbf{$*kk$}, which are the section of the unit sphere that are between two longitudes that are $\frac{\pi}{k}$ apart. Note that $*11$ is the same as $*$ since the corner angles are $\pi$. In fact, every corner with an angle $\pi$ can be removed from the list of corners. Thus, in the orbifold notation, we require every number to be at least $2$ when there are at least two walls.

For the triangular spherical orbifolds, i.e., three mirrors, there are two sub-types. The first sub-type has the form \textbf{$*22k$} where $k>1$. Figure~\ref{fig:non_Euclidean_room} (a) shows one such orbifold (\textbf{$*222$}). This type of orbifolds can be obtained by taking half of the orbifold \textbf{$*kk$} in the northern hemisphere and adding a mirror on the equator. The second sub-type has the form \textbf{$*23k$} where $k=3, 4, 5$. Notice that when $k=6$ we have \textbf{$*236$}, a Euclidean orbifold. From the discussion, we can see there are more spherical orbifolds than Euclidean orbifolds.

There are {\em bad} orbifolds, namely, \textbf{$*k$} where $k>1$ and \textbf{$*k_1k_2$} where $k_2>k_1>1$. Note that neither type of the bad orbifolds can be realized because it is not physically possible to have one great circle intersecting itself at an angle not equal to $\pi$, nor is it possible to have two different great circles that intersect at different angles where they meet. In our system, we do not construct bad orbifolds.

\begin{table}[tb]
	\scriptsize%
	\centering%
	\begin{tabularx}{\columnwidth}{%
			r%
			l l l l l%
		}
		\toprule
		$N$ & Spherical & Euclidean & Hyperbolic  \\
		\midrule
		1 & \textbf{$*$}  &   &   \\
        \hline
		2 & \textbf{$*22$}, \textbf{$*33$}, \textbf{$*44$}, \dots  &   &     \\
        \hline
		3 & \textbf{$*222$}, \textbf{$*223$}, \textbf{$*224$} \dots   &  &   \\
         \hdashline
		 & \textbf{$*233$}, \textbf{$*234$}, \textbf{$*235$} & \textbf{$*236$} & \textbf{$*237$}, \textbf{$*238$}, \textbf{$*239$}, \dots \\
         \hdashline
		 &   & \textbf{$*244$} & \textbf{$*245$}, \textbf{$*246$}, \textbf{$*247$}, \dots\\
         \hdashline
		 &   &   & \textbf{$*2k_2k_3$} ($k_3\ge k_2 > 4$)\\
         \hdashline
		 &   & \textbf{$*333$} & \textbf{$*334$}, \textbf{$*335$}, \textbf{$*336$}, \dots \\
         \hdashline
		 &   &   & \textbf{$*3k_2k_3$} $(k_3 \ge k_2 >3 $)\\
         \hdashline
		 &   &   & \textbf{$*k_1k_2k_3$} $(k_3 \ge k_2 \ge k_1 >3 $)\\
        \midrule
		4 &   & \textbf{$*2222$}  & \textbf{$*k_1k_2k_3k_4$} ($\max_{1\le i\le 4}k_i>2$)  \\
        \hline
		5 &   &    &  \textbf{$*k_1k_2k_3k_4k_5$}    \\
        \hline
		6 &   &    &  \textbf{$*k_1k_2k_3k_4k_5k_6$}    \\
        \hline
		7 &   &    &  \textbf{$*k_1k_2k_3k_4k_5k_6k_7$}    \\
        \hline
		\vdots  &   &    & \multicolumn{1}{c}{\vdots}  \\
		\bottomrule
	\end{tabularx}%
	\caption{Our enumeration of all two-dimensional kaleidoscopic orbifolds based on the cardinality and the universal cover of the orbifolds. The orbifolds on each row have the same $N$, which is the cardinality of the underlying polygon. }
	\label{tbl:polygon_orbifold}
\end{table}

The rest of polygonal kaleidoscopic orbifolds are hyperbolic, and there are no bad hyperbolic orbifolds. There are three cases: (1) three mirrors, (2) four mirrors, and (3) five or more mirrors. An orbifold is hyperbolic if its has five or more mirrors (e.g. Figure~\ref{fig:non_Euclidean_room} (c)). In addition, all quadrangular orbifolds except \textbf{$*2222$} are hyperbolic (e.g. Figure~\ref{fig:non_Euclidean_room} (b)). Finally, triangular hyperbolic orbifolds include six sub-types: (1) \textbf{$*23k$} where $k>6$, (2) \textbf{$*24k$} where $k>4$, (3) \textbf{$*2k_2k_3$} where $k_3\ge k_2>4$, (4) \textbf{$*33k$} where $k>3$, (5)  \textbf{$*3k_2k_3$} where $k_3\ge k_2 > 3$, and (6) \textbf{$*k_1k_2k_3$} where $k_3 \ge k_2 \ge k_1 > 3$. Notice the three cases, each of which corresponds to a Euclidean orbifold that serves as the border between the set of spherical and the set of hyperbolic orbifolds, namely, \textbf{$*236$} for type (1), \textbf{$*244$} for type (2), and \textbf{$*333$} for type (4).

Our enumeration of all two-dimensional kaleidoscopic orbifolds based on the combination of the cardinality of the underlying polygon and the type of its universal cover is shown in Table~\ref{tbl:polygon_orbifold}. We provide the computations behind our enumeration in Appendix~\ref{sec:computation}.

\section{Interactive Orbifold Scene Design}

Our orbifold visualization system consists of two components: a design panel and the display (Figure~\ref{fig:interactive_system}). We employ the Irrlicht game engine~\cite{irrlicht}, which provides an effective balance between interactivity and functionality.

In the design panel, the user can specify the type of the scene by entering its orbifold notation in the form of a number $N$ for the number of walls in the scene and a list of $N$ numbers, $k_1, k_2, ... , k_N$. Here, $k_i$ indicates that the angle of the $i$-th corner is $\frac{\pi}{k_i}$.

The default value of $N$ is five, and five evenly spaced nodes are displayed on the disk inside the design panel (Figure~\ref{fig:interactive_system}), each of which has a default value of two, i.e. \textbf{$*22222$}. The user can change the value of each node, which can be a non-integer in order to create non-orbifold scenes (Figure~\ref{fig:str_curved_rays} (c)). The user can also change $N$, which results in a room with more or fewer walls. The default value for each node in the new setting is again two. Recall that there are two cases that are not physically realizable: (1) a circular room with a single mirror ($N=1$) that self-intersects at an angle not equal to $\pi$ (\textbf{$*k$} where $k>1$), and (2) a room with two mirrors whose two intersection angles are mutually distinct (\textbf{$*k_1k_2$} where $k_2>k_1$). Thus, we disallow these cases from occurring during the design phase. For example, when $N=1$, the value of the only node is set to one and cannot be changed. Similarly, when $N=2$, if the user changes the value of one node, the value of the other node is automatically updated to match it.

Given the orbifold notation, our system instantaneously generates an empty room (a {\em right} polygonal prism) whose floor and ceiling are congruent to the orbifold and whose walls are the sides of the prism. In our system, it is possible to have multiple mirrors on a wall as shown in Figure~\ref{fig:teaser}. The user can also change the height of the room, the color and attenuation of a mirror, and the textures for the ceiling and the floor. Objects can be added to the scene, whose locations, orientations, sizes, base colors, and material properties (e.g. marble, glass) can be modified from their default values as needed. Light sources can also be added to the scene, with control over their locations, shapes, and optical properties. Unwanted objects and light sources can be removed from the scene.

All of the above scene design operations are interactively rendered in order to support the {\em What-You-See-Is-What-You-Get} (WYSIWYG) paradigm, and all of the examples included in the paper and accompanying video were created using our design system.

At the core of our system is the ability to create a room given any arbitrary two-dimensional kaleidoscopic orbifold and to correctly deform an object in the scene when it moves. In addition, in our design system, mirrors are not explicitly generated. Instead, we emulate the mirror effects by creating copies of the original room which together approximate the universal cover of the orbifold. We provide detail on each of these topics next.

\subsection{Orbifold Scene Construction}
\label{sec:construction}

We first compute the Euler characteristic of the orbifold $O$. If $\chi(O)=0$, i.e. a Euclidean orbifold, it is then \textbf{$*2222$}, \textbf{$*333$}, \textbf{$*236$}, or \textbf{$*244$}. As we know the internal angles and the ratios between the side lengths, we can place the vertices of the underlying triangle or square on the floor. Then, with a user-specified room height, we create the ceiling polygon by duplicating the floor polygon and raising it to match the height. Both the floor and the ceiling are represented by a triangular mesh, so are the rectangular walls in the room.

For a non-Euclidean orbifold, its universal cover is either the sphere or the hyperbolic space. Constructing a 3D room over the sphere and the hyperboloid would require a second sphere or hyperboloid to hold the ceiling. While it is possible to construct the room this way, we instead choose to express the orbifold using a planar model, i.e. the stereographic projection for the sphere~\cite{hitchman2009geometry} and the Poincar\'e disk~\cite{coxeter1998non} for the hyperbolic space. By using these models, we have a unified framework in which any polygon, regardless of the type of its universal cover, can be constructed in the plane. Next, we describe our algorithm to identify the side lengths of any non-Euclidean kaleidoscopic orbifold.

\begin{figure}[t]
	\centering%
	$\begin{array}{@{\hspace{0.0in}}c}
	\includegraphics[width=2.5in]{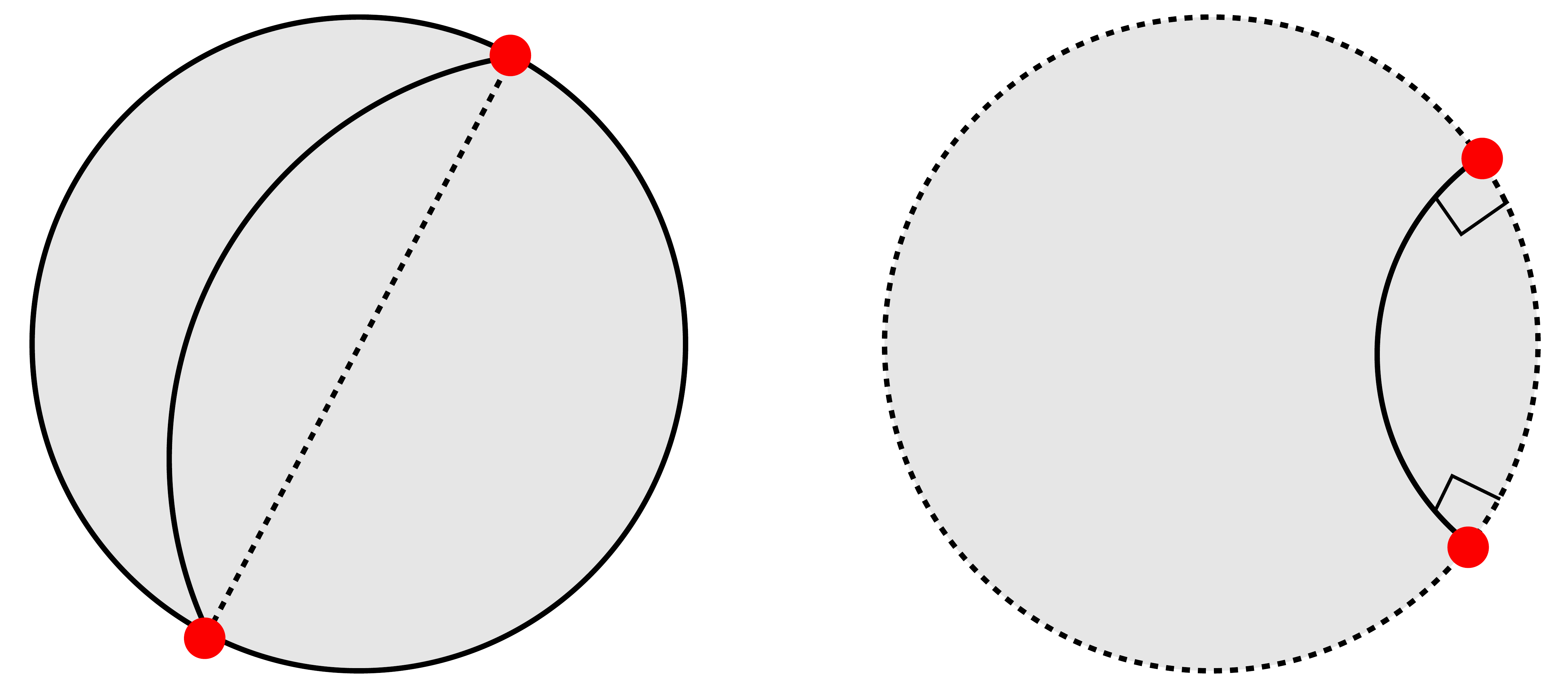} \\
	\end{array}$
	\caption[]{Under the stereographic projection, a geodesic in the sphere is mapped to a circle that intersects the boundary of the unit disk at a pair of antipodal points (left). Using the Poincar\'e disk model, a geodesic in the hyperbolic space is mapped to a circular arc that interests the boundary of the disk at the right angle (right).   }
	\label{fig:geodesics}
\end{figure}

{\bf Spherical Orbifolds:} The stereographic projection~\cite{hitchman2009geometry} maps the unit sphere to the plane $z=0$ such that the equator (a unit circle) is mapped to itself and the north pole is mapped to the origin in the plane. In this case, the northern hemisphere is mapped to the inside of the unit disk bounded by the equator while the southern hemisphere is mapped to the outside of the unit disk. This plane can be identified as the complex plane, i.e. the set of complex numbers. The south pole is mapped to $\infty$. The geodesics are mapped to circles in the plane that intersect the unit circle at a pair of antipodal points (Figure~\ref{fig:geodesics} (left)).

As shown in Table~\ref{tbl:polygon_orbifold}, the underlying polygon of a spherical orbifold is either a monogon, a diangle, or a triangle. In all of these cases, $O$ can be contained inside a hemisphere. That is, under the stereographic projection, it can be contained inside the unit disk in the complex plane. For the monogon, i.e. \textbf{$*$}, the equator is the mirror. For a diangular orbifold \textbf{$*kk$} ($k>1)$, the corner points are on the real axis which are connected by a pair of circle segments that intersect at the corner points at $\frac{\pi}{k}$. Note that the stereographic projection is conformal~\cite{conway2018geometry}, i.e. angle-preserving, thus our choice of the angles at the intersection points. The triangular orbifold \textbf{$*22k$} ($k>1$) is exactly half of the orbifold \textbf{$*kk$} for the same $k$ (Figure~\ref{fig:s2_polygon}: top row).

Finally, the orbifolds \textbf{$*23k$} ($k=3,4,5$) can be constructed by placing the corner points (Figure~\ref{fig:s2_polygon}, bottom row) in the stereographic plane as follows. Let $p_1$, $p_2$, $p_3$ be the corners of the orbifold corresponding to $2$, $3$ and $k$, respectively. The spherical lengths, $d_{i, i+1}$, of the edges $p_i p_{i+1}$ in the polygon are uniquely determined by the angles $\frac{\pi}{k_i}, \frac{\pi}{k_{i+1}}$ and $\frac{\pi}{k_{i+2}}$ of the polygon as follows~\cite{thurston1997three}:

\begin{equation}
d_{i, i+1}=\cos^{-1}\left(\frac{\cos \left(\frac{\pi}{k_{i+2}}\right)+\cos \left(\frac{\pi}{k_{i}}\right) \cos \left(\frac{\pi}{k_{i+1}}\right)}{\sin \left(\frac{\pi}{k_{i}}\right) \sin \left(\frac{\pi}{k_{i+1}}\right)} \right)
\label{eq:dis_s2_edge}
\end{equation}

\noindent where $i = 1, 2, 3$, $d_{3,4}=d_{3,1}$, $k_{4}=k_{1}=2$, $k_{5}=k_{2}=3$, $k_3=k$, and $p_4 = p_1$.
With this information, we first place the vertex $p_1$ at $(1,0)$ in the complex plane. Next, we compute a geodesic emanating from $p_1$, along which we travel for a {\em spherical} distance of $d_{1,2}$ to find $p_2$.
Since $p_1$ and $p_2$ are both represented as complex numbers in the stereographic plane, their spherical distance can be computed from their complex number representations as follows~\cite{hitchman2009geometry}:

\begin{equation}
d(p_1, p_2)=2\tan^{-1}\left(\left|\frac{p_2-p_1}{1+\overline{p_1}p_2}\right|\right)
\label{eq:dis_s2}
\end{equation}

\noindent where $\overline{p}$ is the conjugate of a complex number $p$. Solving for $p_2$ in Equation~\ref{eq:dis_s2} can be challenging given any arbitrary geodesic $\gamma$. However, on the unit circle in the stereographic plane, one can find $p_2$ without the need to solve Equation~\ref{eq:dis_s2}. This is because the unit circle corresponds to the equator in the sphere under the stereographic projection, thus the spherical distance between the two points is the same as the arc distance between them on the unit circle. However, $\gamma$ is not always the unit circle. To address this, we identify a translation in the sphere, which, under the stereographic projection, takes $\gamma$ (Figure~\ref{fig:tr_s2_h2} (left): the arc) to the upper-half of the unit circle. Such a translation
can be modelled by M\"obius transformations~\cite{basmajian2006mobius} that have the following form:

\begin{equation}
  f(z) = e^{i\theta}\frac{z-z_0}{\overline{z_0}z+1}.
\label{eq:mobius}
\end{equation}

\noindent Here, $\theta \in [0, 2\pi)$ and $z_0$ and $z$ are complex numbers. A M\"obius transformation is uniquely determined by three pairs of corresponding points~\cite{hitchman2009geometry}. We first extend $\gamma$ until it intersects the unit disk and map the intersection points to $(-1, 0)$ and $(1, 0)$, respectively (Figure~\ref{fig:tr_s2_h2} (left)). We select the third point to be middle point of $\gamma$, which is mapped to $(0, 1)$. Call this M\"obius transformation $M$. Since M\"obius transformations in the stereographic plane correspond to rotating the sphere before the stereographic projection, spherical distance is preserved $M$. Thus, we can find the point $M(p_2)$ on the unit circle which has an arc distance of $d_{1,2}$ from $M(p_1)$. Finally, performing the inverse of $M$, which is also a M\"obius transformation to $M(p_2)$, we can find $p_2=M^{-1}(M(p_2))$.

\begin{figure}[t]
	\centering%
	$\begin{array}{@{\hspace{0.0in}}c}
	\includegraphics[width=3.3in]{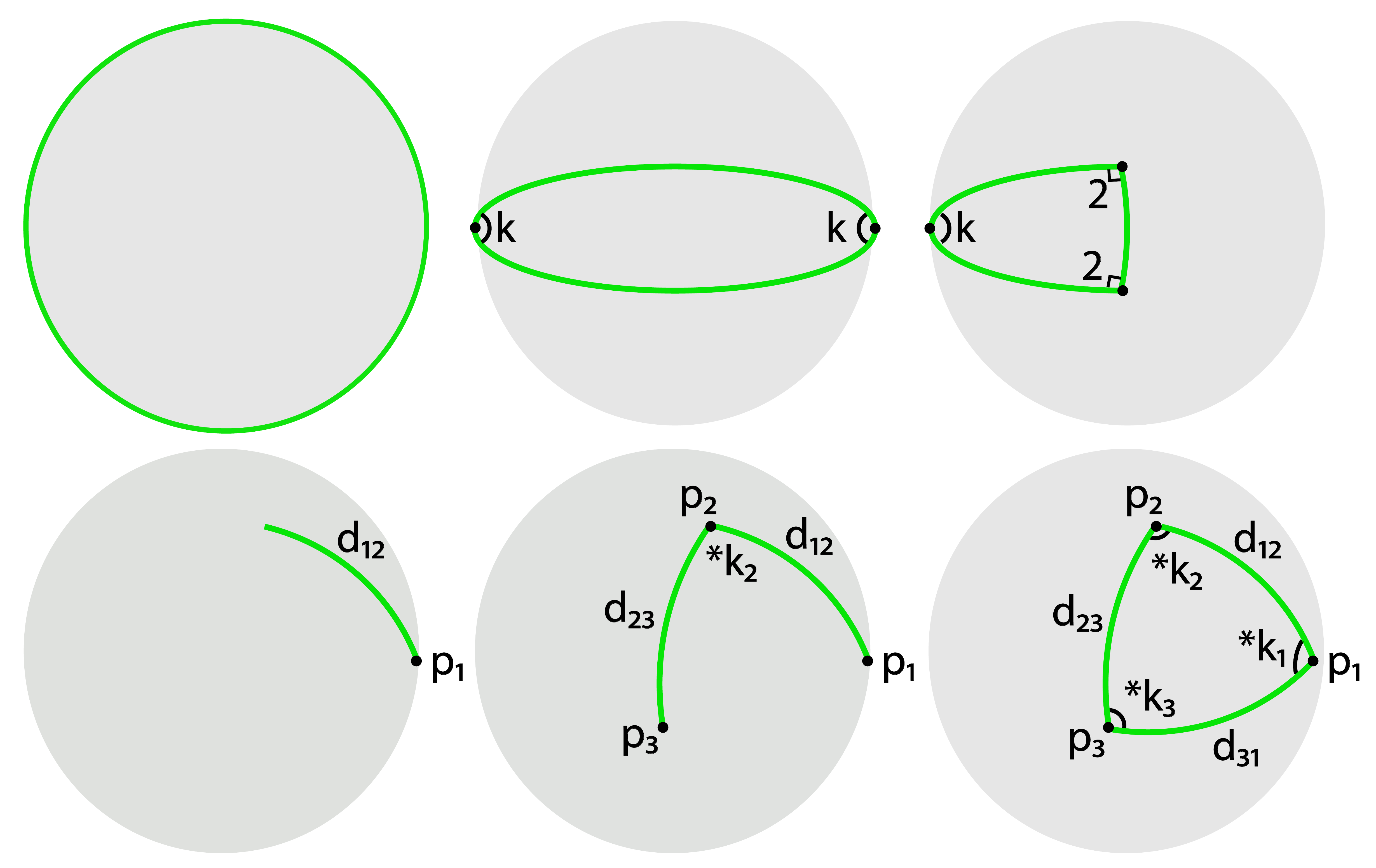} \\
	\end{array}$
	\caption[]{In the stereographic plane, the \textbf{$*$} is realized as the unit disk (top row: left). A \textbf{$*kk$} orbifold can be created by placing its corner points at a pair of antipodal points on the unit disk (top row: middle). The \textbf{$*22k$} type of orbifold is half of the \textbf{$*kk$} orbifold (top row: right). The \textbf{$*23k$} ($k=3, 4, 5$) can be constructed as shown in the bottom row.}
	\label{fig:s2_polygon}
\end{figure}

From $p_2$, we compute a second geodesic which has an angle of $\frac{\pi}{k_2}$ with the first geodesic. Then, travelling along the new geodesic for a {\em spherical} distance of $d_{2,3}$, we can locate $p_3$. Thus, we have constructed the \textbf{$*23k$} type of orbifolds (Figure~\ref{fig:s2_polygon}: bottom row).

\begin{figure}[h]
	\centering%
	$\begin{array}{@{\hspace{0.0in}}c}
	\includegraphics[width=2.7in]{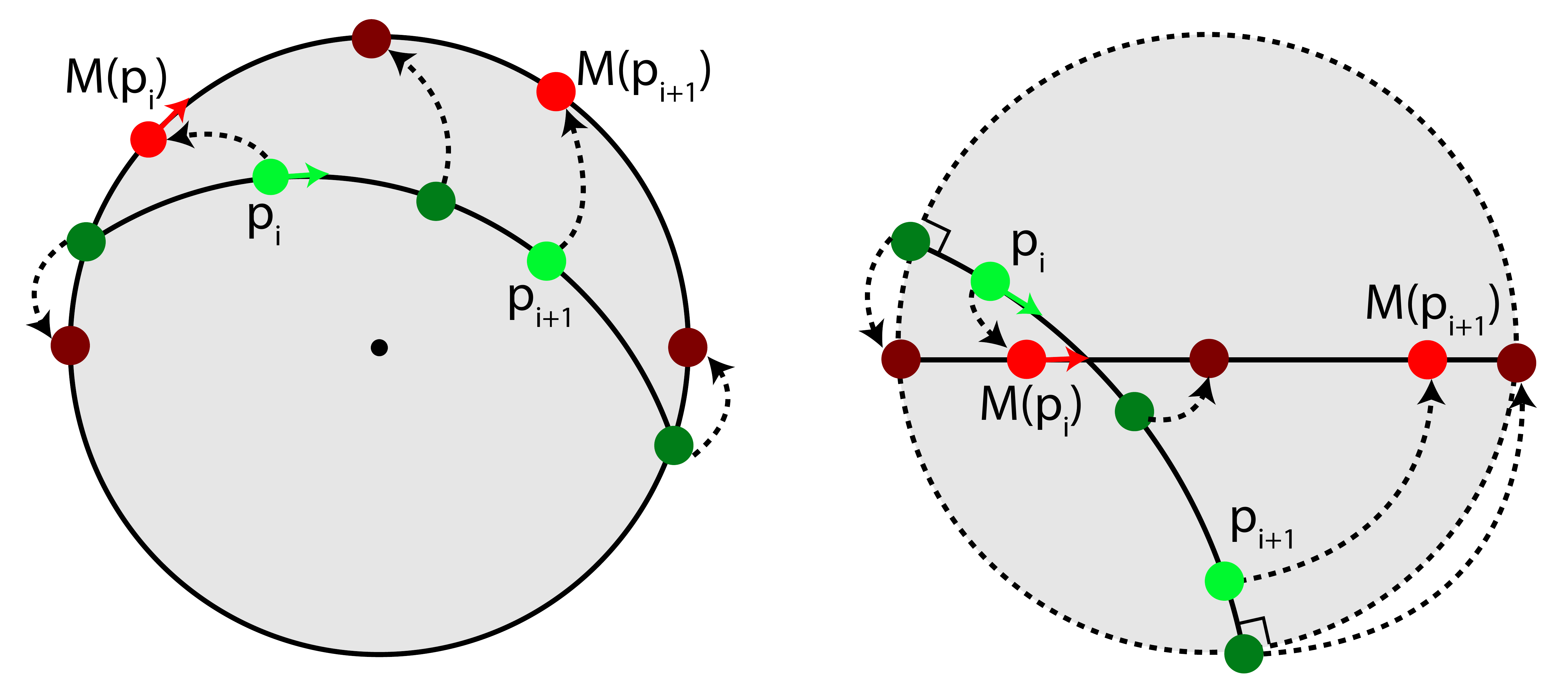} \\
	\end{array}$
	\caption[]{Given a point $p_i$, we simplify the computation of $p_{i+1}$ by performing a M\"obius transformation in the stereographic plane for spherical orbifolds (left) and the Poincar\'e disk for hyperbolic orbifolds (right). In both cases, the unique M\"obius transformation $M$ maps the extended geodesic (including the intersections with the unit disk) to the upper semi-circle for spherical orbifolds and line segment between $(-1, 0)$ and $(1, 0)$ for hyperbolic orbifolds. In addition, the center of extended geodesic is mapped to $(0, 1)$ in the spherical case (left) and $(0, 0)$ in the hyperbolic case (right). Then we identify $M(p_{i+1})$ from which we can recover $p_{i+1}$ using the inverse of $M$.
	}
	\label{fig:tr_s2_h2}
\end{figure}


{\bf Hyperbolic Orbifolds:} We now consider the hyperbolic case, where $\chi(O)<0$. The hyperbolic space can be modelled by the Poincar\'e disk~\cite{coxeter1998non}, which is the interior of the unit disk in the plane. Under this model, the geodesics of the hyperbolic space are circles that intersect the boundary of unit disk at the right angle (Figure~\ref{fig:geodesics} (right)). Recall that a hyperbolic orbifold is a polygon with three or more sides (Table~\ref{tbl:polygon_orbifold}). To create such a polygon, we follow the same approach for spherical orbifolds. That is, we start with the location of $p_1$ in the Poincar\'e disk and a geodesic emanating from $p_1$. We then travel along this geodesic for a prescribed distance to locate $p_2$. From there, we identify a new geodesic whose angle with the original geodesic is $\frac{\pi}{k_2}$, which we follow to identify $p_3$. This process terminates once when we have identified $p_{N}$. The main difference lies in the fact that the hyperbolic distance between two points in the Poincar\'e disk is different from the spherical distance of the same two points in the stereographic plane.

Given a triangular hyperbolic orbifold \textbf{$*k_1 k_2 k_3$}, the length of the edge between $p_i$ and $p_{i+1}$ ($i = 1, 2, 3$, $d_{3,4}=d_{3,1}$, $k_{4}=k_{1}$, $k_{5}=k_{2}$, and $p_4 = p_1$) is given by~\cite{thurston1997three}
\begin{equation}
d_{i,i+1} = \cosh^{-1} \left( \frac{ \cos \left(\frac{\pi}{k_{i+2}}\right) + \cos \left(\frac{\pi}{k_i}\right) \cos \left(\frac{\pi}{k_{i+1}}\right) }{
	\sin \left(\frac{\pi}{k_i}\right) \sin \left(\frac{\pi}{k_{i+1}}\right) } \right)\\
\label{eq:h2_tri1}
\end{equation}

In addition, when represented as complex numbers in the plane containing the Poincar\'e disk, the hyperbolic distance between $p_i$ and $p_{i+1}$ is given by~\cite{hitchman2009geometry}:

\begin{equation}
d(p_i, p_{i+1})=\ln \left(\frac{|1-\overline{p_i} p_{i+1}|+|p_{i+1}-p_i|}{|1-\overline{p_i} p_{i+1}|-|p_{i+1}-p_i|}\right)
\label{eq:dis_h2}
\end{equation}

Similar to the case of spherical orbifolds, finding $p_{i+1}$ from $p_i$ on an arbitrary geodesic $\gamma$ in the Poincar\'e disk requires solving Equation~\ref{eq:dis_h2} for $p_{i+1}$ which can be challenging. To simplify the matter, we identify the hyperbolic translation in the Poincar\'e disk that takes $\gamma$ (Figure~\ref{fig:tr_s2_h2} (right): the arc) to the line segment from $(-1, 0)$ to $(1, 0)$ in the Poicnar\'e disk (Figure~\ref{fig:tr_s2_h2} (right)). This translation  takes the intersection with the unit circle (two dark green points on the unit circle) to $(1, 0)$ and $(-1, 0)$ (the rightmost and the leftmost burgundy points). It also takes the middle points on $\gamma$ (the middle dark green point) to $(0, 0)$ (the middle burgundy point). Such a translation can be modelled by a M\"obius transformation of the following form:

\begin{equation}
 f(z) = e^{i\theta}\frac{z-z_0}{1-\overline{z_0}z}
\label{eq:mobius_hyperbolic}
\end{equation}

\noindent where $\theta \in [0, 2\pi)$ and $|z_0|<1$. Call the translation $M$. Since translations maintain hyperbolic distances, the hyperbolic distance between $p_i$ and $p_{i+1}$ is the same as the distance between $M(p_i)$ and $M(p_{i+1})$. However, since $M(p_i)$ and $M(p_{i+1})$ are on the real axis, it is easier to solve Equation~\ref{eq:dis_h2}.  Once we have found $M(p_{i+1})$, we can recover $p_{i+1}$ by applying $M^{-1}$.

However, deciding the side lengths of a hyperbolic polygon with at least four edges is more challenging as there are no published formulas to the best of our knowledge. To address this, we compute the side lengths based on two facts~\cite{Milley:98}: (1) any quadrangular hyperbolic polygon can be decomposed into the disjoint union of at most two quads with two right angles (\textbf{$*22k_3k_4$}), and (2) any hyperbolic polygon with at least five sides can be decomposed into the disjoint union of a finite number of quads of the type \textbf{$*22k_3k_4$} and pentagons of the type \textbf{$*2222k_5$}, i.e. four right angles. Examples of the two facts are shown in Figure~\ref{fig:hyper_decomp}.

\begin{figure}[t]
	\centering%
	$\begin{array}{@{\hspace{0.0in}}c}
	\includegraphics[width=3.0in]{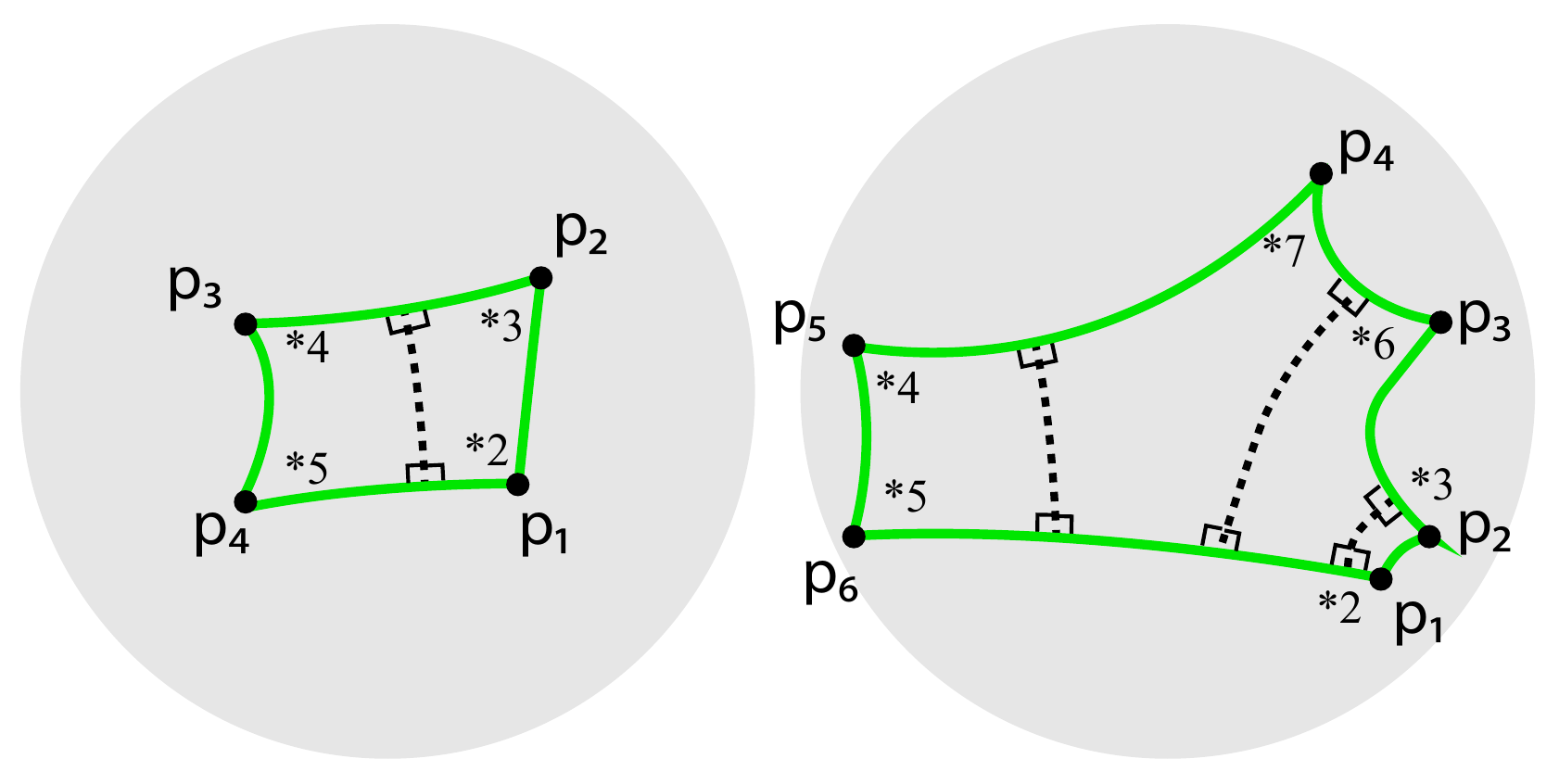} \\
	\end{array}$
	\caption[]{Two example scenarios of our decomposition algorithm for hyperbolic orbifolds: (left) a quad is divided into two \textbf{$*22k_3k_4$} type quads, and (right) a hexagon is divided into two \textbf{$*22k_3k_4$} quads and two \textbf{$*2222k_5$} pentagons.  }
	\label{fig:hyper_decomp}
\end{figure}

The side lengths of the \textbf{$*22k_3k_4$} polygon (Figure~\ref{fig:right_quad_pentagon} (left)) are given by~\cite{thurston1997three, Milley:98}:

\begin{equation}
\begin{aligned}
& d_{2,3} = \sinh^{-1}{\left(\frac{\cos\left({\frac{\pi}{k_4}}\right) + \cos\left({\frac{\pi}{k_3}}\right)\cosh\left({d_{1,2}}\right)}{\sin\left({\frac{\pi}{k_3}}\right) \sinh\left({d_{1,2}}\right)} \right)}  \\
& d_{3,4} = \cosh^{-1}{\left(\frac{\cosh\left({d_{1,2}}\right) + \cos\left({\frac{\pi}{k_3}}\right)\cos\left({\frac{\pi}{k_4}}\right)}{\sin\left({\frac{\pi}{k_3}}\right) \sin\left({\frac{\pi}{k_4}}\right)} \right)} \\
& d_{4,1} = \sinh^{-1}{\left(\frac{\cos\left({\frac{\pi}{k_3}}\right) + \cos\left({\frac{\pi}{k_4}}\right)\cosh\left({d_{1,2}}\right)}{\sin\left({\frac{\pi}{k_4}}\right) \sinh\left({d_{1,2}}\right)} \right)}
\end{aligned}
\label{eq:hyper_quad}
\end{equation}

\noindent where $d_{1,2}$, the length of the cut edge in the decomposition, is a free variable. This is similar to the case where the width and length of the \textbf{$*2222$} orbifold (a rectangle) are free variables.

A generic quadrangular orbifold (\textbf{$*k_1 k_2 k_3 k_4$}) can be decomposed into the disjoint union of two quads (\textbf{$*22k_1k_2$}) and (\textbf{$*22k_3k_4$}) (Figure~\ref{fig:hyper_decomp}, left). This allows us to compute side lengths of the two special quads, which, when combined, give rise to the side lengths of the generic quad.

The side lengths of the \textbf{$*2222k_5$} polygon (Figure~\ref{fig:right_quad_pentagon} (right)) can computed as follows:~\cite{thurston1997three, Milley:98}:

\begin{equation}
\begin{aligned}
& d_{1,2} = \cosh^{-1}{\left(\frac{\cos\left({\frac{\pi}{k_5}}\right) + \cosh\left({d_{2,3}}\right)\cosh\left({d_{5,1}}\right)}{\sinh\left({d_{2,3}}\right) \sinh\left({d_{5,1}}\right)}\right)}  \\
& d_{4,5} = \sinh^{-1}{\left(\frac{\cosh\left({d_{2,3}}\right) + \cos\left({\frac{\pi}{k_5}}\right)\cosh\left({d_{5,1}}\right)}{\sinh\left({d_{5,1}}\right) \sin\left({\frac{\pi}{k_5}}\right)}\right)} \\
& d_{3,4} = \sinh^{-1}{\left(\frac{\cosh\left({d_{5,1}}\right) + \cos\left({\frac{\pi}{k_5}}\right)\cosh\left({d_{2,3}}\right)}{\sinh\left({d_{2,3}}\right) \sin\left({\frac{\pi}{k_5}}\right)}\right)}
\end{aligned}
\label{eq:hyper_penta}
\end{equation}

\noindent where $d_{2,3}$ and $d_{5,1}$, the cut edges in the decomposition, are free variables.

Generic pentagons and higher-order $N$-sided polygons can be decomposed as the union of $N-4$ pentagons of the type \textbf{$*2222k_5$} with up to two additional quads of the type \textbf{$*22k_1k_2$}. To do so, we first consider the simplest case where there exist $k_i=k_{i+1}=2$ and $k_j=k_{j+1}=2$ where $i, i+1, j, j+1$ are mutually distinct. We can find a geodesic that intersects $p_{i-1}p_i$ and $p_{i+1}p_{i+2}$ at the right angle. This geodesic removes from the original polygon a pentagon involving $p_i$, $p_{i+1}$, and $p_{i+2}$ which has four right internal angles, thus the type \textbf{$*2222k_5$}. The remaining polygon has one fewer vertex and still has four internal right angles (Figure~\ref{fig:hyper_decomp}). Repeating this process can lead to $N-4$ pentagons of the type \textbf{$*2222k_5$}. On the other hand, a generic polygon with at least five edges can be reduced into the simpler setting by removing up to two quads of the type \textbf{$*22k_3k_4$}. We can then compute the side lengths of each of the sub-polygons, which, when combined, give the side lengths of the original polygon.

\begin{figure}[t]
	\centering%
	$\begin{array}{@{\hspace{0.0in}}c}
	\includegraphics[width=3.0in]{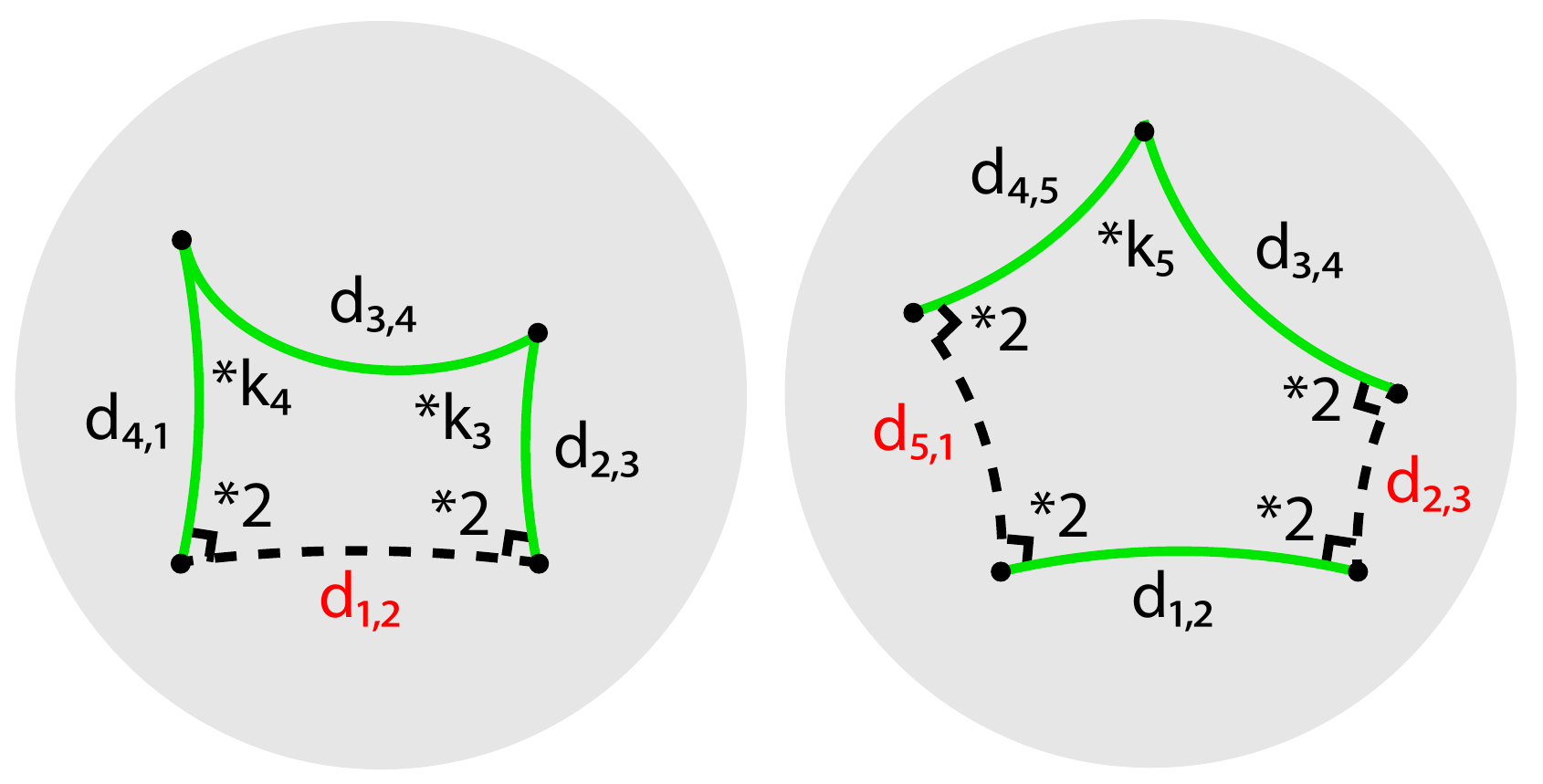} \\
	\end{array}$
	\caption[]{The special quad \textbf{$*22k_3k_4$} (left) and the special pentagon \textbf{$*2222k_5$} (right) are the building blocks of our decomposition algorithm. The free variables are colored in red, which correspond to the edges (dashed) introduced during the decomposition.     }
	\label{fig:right_quad_pentagon}
\end{figure}

Every cut edge in the decomposition gives rise a free variable, which can be modified by the user. The default value for the free variables is set to $1.4$.

With our algorithm, any spherical and hyperbolic orbifold can be constructed given its orbifold notation. Figure~\ref{fig:2d_pattern} show two example orbifolds with their universal covers.

\subsection{Object Embedding and Movement}

Once the orbifold has been realized as a polygonal prism, the user can add objects to the scene.

However, when bringing an object, which is created in a presumably Euclidean space, to a non-Euclidean space, a natural question to ask is how to perform the embedding. Due to the difference in their respective distance metrics, it is not always possible to embed the model in such a way that the length of every edge in the mesh is maintained. To address this challenge, when embedding the object into the scene, we first place it so that its center of mass is at the origin of the plane for both the stereographic plane and the Poincar\'e disk. The coordinates of the object are now considered their corresponding coordinates in the stereographic plane and the Poincar\'e disk. Then, the embedded mesh is translated to the user-specified initial location with the translation native to the non-Euclidean space.

Translations in both the spherical and hyperbolic spaces are isometries. In the spherical space, translations can be modelled in the stereographic plane by M\"obius transformations according to Equation~\ref{eq:mobius}. Interestingly, translations in the hyperbolic space using the Poincar\'e disk can also be modelled by M\"obius transformations according to Equation~\ref{eq:mobius_hyperbolic}.

We store the M\"obius transformation of each object, and update it when the model is interactively moved inside the room. The M\"obius transformation, which corresponds to a translation of the spherical space or the hyperbolic space, is applied to all the vertices in the mesh.

\begin{figure}[t]
	\centering%
	$\begin{array}{c@{\hspace{0.1in}}c}
	\includegraphics[width=1.6in]{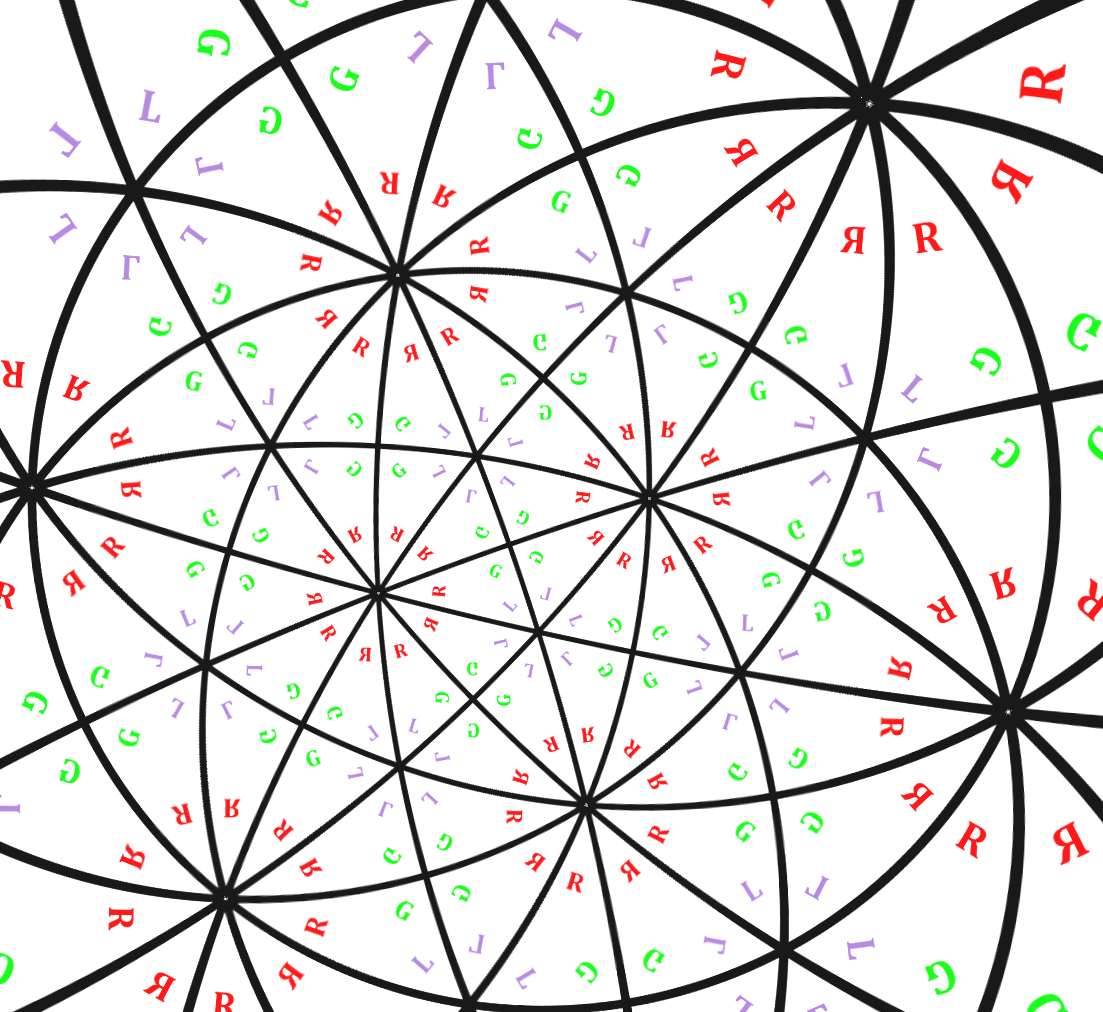}
	&\includegraphics[width=1.6in]{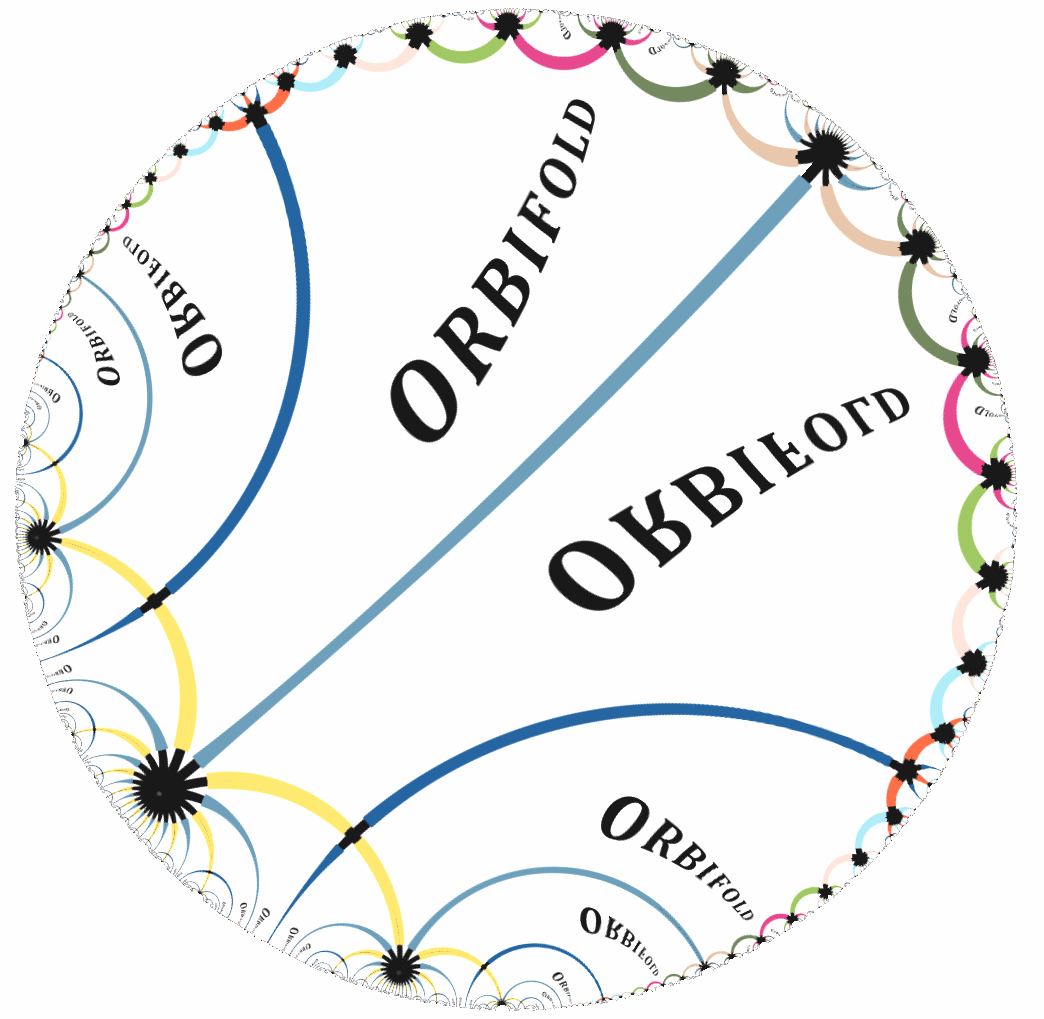}
	\\
	(a) \quad \textbf{$*235$}
	& (b) \quad \textbf{$*23456789(10)(11)$}\\
	\end{array}$
	\caption[]{Two orbifolds generated using our orbifold layout creation algorithm: (a) the $*235$ orbifold (spherical) and (b) the orbifold $*23456789(10)(11)$ (hyperbolic). }
	\label{fig:2d_pattern}
	\end{figure}

\subsection{Universal Cover Construction}

Being able to see the universal cover, i.e. all mirror reflections, can be important for a user while exploring our tool. However, capturing mirror reflections can be computationally expensive with high-quality renderers such as Mitsuba~\cite{Mitsuba}. Williams and Zhang~\cite{Williams:16} address this for Euclidean kaleidoscopic orbifolds by creating a finite number of copies of the reflections of the original room that approximate the universal cover. This is based on two observations. First, as copies are farther away from the room, their images perceived by the viewer approach the vanishing line and thus do not contribute much to the pixels. Second, the color intensity of faraway copies diminishes as the number of bounces off from the mirrors increases. We employ a similar approach, which focuses on non-Euclidean orbifolds. Two examples of universal covers created by our method are shown in Figure~\ref{fig:2d_pattern}: (a) the spherical orbifold \textbf{$*235$} and (b) the hyperbolic orbifold \textbf{$*23456789(10)(11)$}.

In our system, the construction of the universal cover for Euclidean orbifolds closely follows that of Williams and Zhang~\cite{Williams:16}, which computes the translational cover of the orbifold and generates additional copies of the translational cover using either the Gaussian integer lattice $\mathbb{Z}[\textit{i}]$ for the \textbf{$*2222$} and \textbf{$*442$} cases and the Eisenstein integer lattice $\mathbb{Z}[\omega]$ for the \textbf{$*333$} and \textbf{$*632$} cases.

For spherical and hyperbolic orbifolds, the notion of translational cover is not well-defined. Consequently, we employ the following process. Starting from the original room, we iteratively add a virtual copy by reflecting the room across one of its mirrors. It is also possible to reflect a virtual room across its mirror. To avoid duplicates, i.e. a virtual room that is discovered through two different paths from the original room, we compare the center of a potentially new room to the centers of already visited room and virtual rooms~\cite{zeller2021tegula}.

To locate the corners of each newly added room, we apply reflections in the spherical or hyperbolic spaces to the already visited room involved in the reflection. Inside the stereographic plane, a reflection in the spherical space can be represented as the compositions of some M\"obius transformations (Equation~\ref{eq:mobius}) and the conjugation function with respect to the real axis $f(z)=\overline{z}$. Similarly, inside the Poincar\'e disk, a reflection in the hyperbolic space can also be expressed as the composition of some M\"obius transformation (Equation~\ref{eq:mobius_hyperbolic}) and the conjugation function with respect to the real axis. The M\"obius transformation is stored in the form of a $2\times 2$ matrix whose entries consist of $a$, $b$, $c$, and $d$~\cite{anderson2006hyperbolic}.

For each copy of the room we save the transformation needed to take the original room to the copy, which is a combination of a M\"obius transformation and up to one conjugation function. In addition, for each object in the original room, we also save its M\"obius transformation. Then, the position and orientation of an object in a virtual copy of the original scene can be calculated by combining the matrix for the virtual room and the matrix of the same object in the original scene. Also, the computation for both object manipulation and universal cover construction is done using the shader.

\section{Optics-Based Orbifold Visualization}
\label{sec:vis_results}

Our mirror-based visual metaphor and interactive orbifold design algorithm are motivated by a number of tasks for understanding the concept of orbifolds. Additional results illustrating concepts and properties related to orbifolds can be found in Appendix~\ref{sec:more_visualization_results}. Below we list four tasks and show the results of using our approach to address them.

{\bf Recognizing Non-Orbifolds}: One of the most fundamental tasks for orbifold visualization is to decide whether a given mirror scene is an orbifold. Non-orbifolds can be difficult to conceptualize, especially when one or more angles at the corners violates $\frac{\pi}{k}$ for $k\in \mathbb{N}^+$. The resulting rendered images and/or videos can deliver a sense of whether the room is a kaleidoscopic orbifold. When the room is a kaleidoscopic orbifold, the viewer should be able to move around the scene and see a consistent larger scene (the universal cover) that is seamlessly tiled by copies of the original room (the orbifold). On the other hand, if the viewer sees conflicting results, such as the double-headed dragon in Figure~\ref{fig:str_curved_rays} (c), it is an indication of a non-orbifold. In this particular example, the room has a corner with an angle of $\frac{\pi}{k}$ where $k=\frac{2}{3}$, which is not an integer.

\begin{figure*}[t]
	\centering%
	$\begin{array}{@{\hspace{0.0in}}c@{\hspace{0.1in}}c@{\hspace{0.1in}}c}
	\includegraphics[width=2.3in]{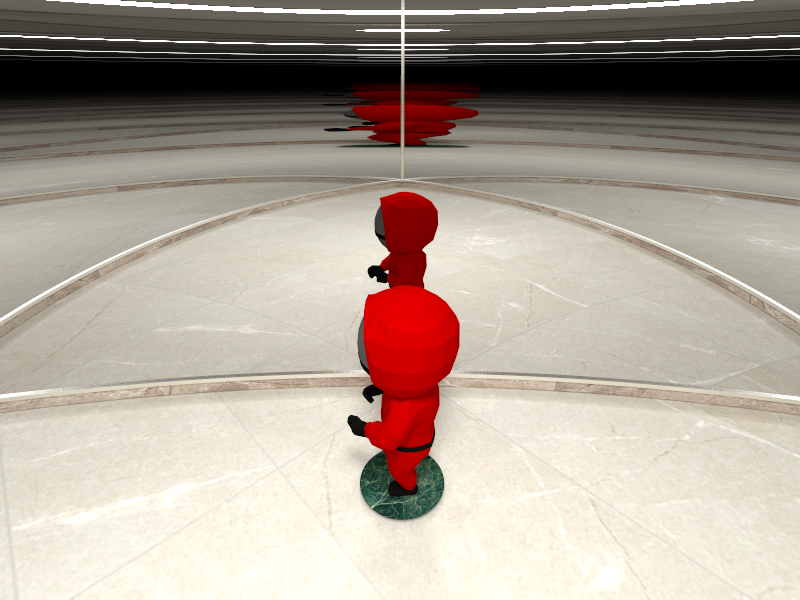}
	&\includegraphics[width=2.3in]{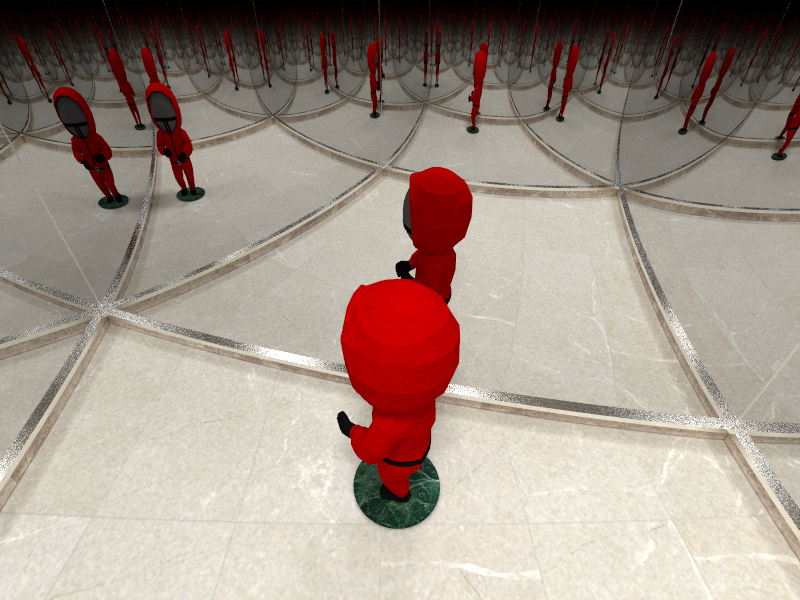}
	&\includegraphics[width=2.3in]{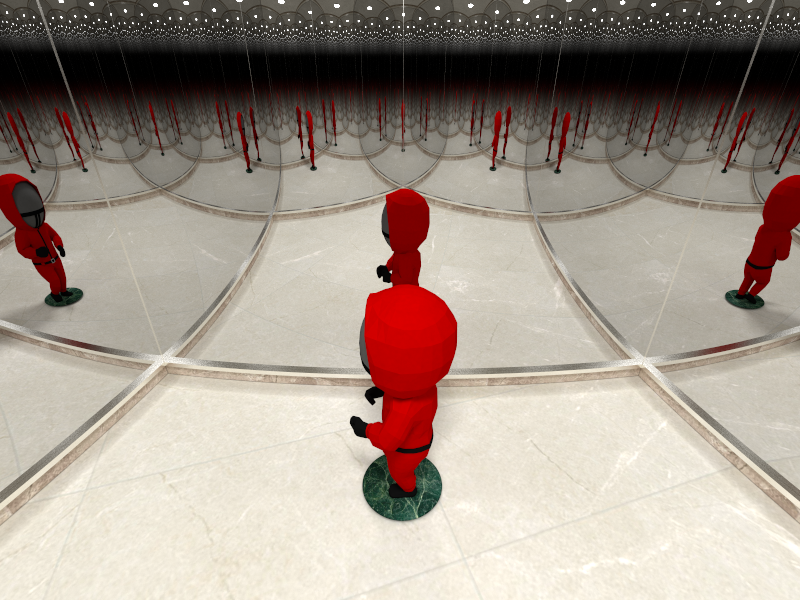} \\
	(a) \quad \textbf{*222} & (b) \quad \textbf{*2323} & (c) \quad \textbf{*22222}  \\
	\end{array}$
	\caption[]{An avatar walks along the walls of the orbifolds. The universal cover is the sphere for (a) and the hyperbolic space for (b) and (c).}
	\label{fig:non_Euclidean_room}
\end{figure*}

{\bf Sensing the Orbifold}: An orbifold is a topological space that stands on its own. Using our mirrored scenes, we can produce visualizations that emphasize different aspects of an orbifold. For example, for the scene shown in Figure~\ref{fig:emphasis}, we can adjust the attenuation factors for the mirrors to emphasize the orbifold itself (Figure~\ref{fig:emphasis} (a)) by setting high attenuation for both mirrors, or the translational cover (Figure~\ref{fig:emphasis} (b)) by setting high attenuation for one mirror and low attenuation for the other mirror, or the universal cover (Figure~\ref{fig:emphasis} (c)) by setting low attenuation for both mirrors. Notice that in all of these cases, our mirror-based visualization provides a clear sense of the orbifold, the room.

{\bf Determining the Orbifold Notation}: Once an intuitive sense of the orbifold is established and the structure of the reflection is observed, it becomes a relatively straightforward task to determine the type of the orbifold in terms of its orbifold notation.
Given the power of 3D graphics and animation, we produce either a panorama of the scene or an animation in which an avatar walks along the walls (see the accompanying video). As the avatar approaches the $i$-th corner, so do its nearest reflections. There appear to be $k_i$ pairs of avatars approaching the corner. Then, $k_i$ pairs of the same avatars leave the corner for the next one. This helps the user identify $k_i$ for the $i$-th corner. By walking around the room one time, each corner is visited. This can help the user write down the orbifold notation.

{\bf Identifying the Type of the Universal Cover}: With the same panorama and walkaround animations, the user can gain insight into the type of the universal cover. Both the spherical space and the hyperbolic space can bring unfamiliar experience to someone new to the concept. For example, inside the spherical space, objects do not always appear smaller when they are further away. Specifically, when the viewer is at the south pole, objects near the north pole appear much wider than the same object at the equator (Figure~\ref{fig:non_Euclidean_room} (a)). On the other hand, objects farther away always look smaller in the hyperbolic space. However, with the seemingly same distance to the viewers, an object can look much smaller in the hyperbolic space than in the Euclidean space (Figure~\ref{fig:non_Euclidean_room} (b-c)). As the objects move around the scene such as the avatar in Figure~\ref{fig:non_Euclidean_room} and the accompanying video, the way the reflections of the objects deform in non-Euclidean spaces is rather different from that in the Euclidean space. For example, a reflection of the avatar may suddenly grow much bigger and then shrink quickly in the spherical space. With a relatively simple scene (the orbifold), the user can gain insight into the entire universal cover through mirror reflections.

\section{Rendering}
\label{sec:curved_rays}

Once the scene has been constructed, we can either display it interactively during the design stage (Figure~\ref{fig:interactive_system}) or send it to Mitsuba~\cite{Mitsuba} for high-quality off-line rendering. For the latter, we make use of a ray-tracing type of approach with Mitsuba. Recall that our non-Euclidean room is modelled by a subset in $\mathbb{D} \times [0, h]$ where $\mathbb{D}$ is the unit disk and $h$ is the height of the room. Thus, any geodesic $\gamma$ in the space satisfies that its horizontal projection is a geodesic in the spherical or hyperbolic space. Consequently, $\gamma$ takes the form of a 3D spiral. To generate the correct rendering of such scenes, we modify the ray-triangle intersection algorithm in Mitsuba from line-plane intersection to spiral-plane intersection.

Given a point $p_0$ in the room and a tangent direction $v$, the spiral emanating from $p_0$ in the direction of $v$ is the tensor product of horizontal projection and the vertical projection of the spiral. The horizontal projection is a circle (can also be a line) in the floor that starts with the projection of $p_0$ onto the floor and travels in the direction of $v_H$, the projection of $v$ onto the floor. The vertical projection of the spiral is a line that starts at $p$ that travels in the direction of $v_Z$, the $z$-component of $v$. This leads to a parametric form of $\gamma$ which we use to find the intersection of the spiral $\gamma$ with a given triangle on an object in the scene. When there are multiple intersections with the triangle, we find the one that has the smallest positive $t$ value, which represents the closest intersection from the reference point $p$.

During the design stage, we make use of scan conversion for interactive feedback. Recall that in this case, we also render the universal cover since we do not model mirror reflections during design. For each vertex on an object in the scene or a virtual copy, we project it to the image plane by using the same spiral-triangle intersection as in the case of Mitsuba rendering. We then perform barycentric interpolation to find the footprint of any triangle in the image plane by interpolating their vertices' locations. Given enough mesh resolution, the error in this interpolation-based approach is relatively small since all of its vertices are projected correctly using the spiral-triangle intersection.

\section{Performance}

Our interactive design system is evaluated on  a computer featuring an i7-8700K @3.70GHz CPU and an NVIDIA GeForce GTX 2080 GPU. For an orbifold $*22222$ with 26,022 primitives at an $800 \times 600$ resolution, our system can render its universal cover approximated by $166$ copies of the original room at $11$ frames per second. When the number of copies is reduced to $61$, the frame rate is increased to $30$ frames per second.

For off-line rendering, we make modifications to Mitsuba $0.6$. The rendering was performed on a Linux Cluster, comprised of machines equipped with $12$-core $2.67$ GHz processors and $48$ GB SDRAM, to generate high-resolution images and videos based on the scenes that we created. Due to the increased computational cost from the spiral-triangle intersections, the rendering is considerably slower for non-Euclidean orbifolds than Euclidean orbifolds. As an example, for the single frame of the spherical orbifold $*222$ that contains $12,572$ primitives (shown in Figure~\ref{fig:non_Euclidean_room} (a)), it required $7.69$ minutes when rendered at a resolution of $800 \times 600$.

\section{Conclusion and Future Work}

In this paper, we propose the use of mirror reflections as a visual metaphor for orbifolds and provide a system in which the user can interactively design any two-dimensional kaleidoscopic orbifold. At the core of our interactive design system is the ability to determine the configuration of the room (locations of the corners, shapes of the walls) given the orbifold notation provided by the user. Our system can handle not only Euclidean orbifolds, but also spherical orbifolds and hyperbolic orbifolds. In addition, we provide an enumeration of two-dimensional kaleidoscopic orbifolds based on the combination of the cardinality of the underlying polygon and the type of the universal cover. As part of the design system, we also enable the interactive construction of the universal cover of the orbifold as well as movement of the objects in the scene. The user can also generate high-quality photorealistic renderings of the scene, panorama, and animations with Mitsuba, which we have modified to account for the geodesics in spherical and hyperbolic geometry along which light travels.

Making rendering more efficient with the off-line rendering with Mitsuba is important, and we plan to investigate efficient spatial hashing data structures for the non-Euclidean spaces. In addition, the quality of the meshes used to represent the floor and the ceiling can impact the rendering speed, and we plan to explore optimal meshing structures for our purpose.

For future directions, we wish to expand our rendering system to arbitrary three-dimensional orbifolds. In addition, not all orbifolds are kaleidoscopic, and we would like to incorporate the visualization for non-kaleidoscopic orbifolds, such as the ones involving gliding reflections. Finally, we wish to explore the visualization of non-orientable orbifolds, whose universal cover is a non-orientable surface such as the Projective plane and the Klein bottle.

\acknowledgments{The authors wish to thank our anonymous reviewers for their constructive feedback. We appreciate the help from Peter Oliver during video production. Botong Qu has provided valuable suggestions during the initial discussion phase. This work was supported in part by the NSF award (\# 1619383).}

\bibliographystyle{abbrv-doi-hyperref}

\bibliography{kaleidoscopes}

\clearpage
\appendix

\section{Proofs for Kaleidoscopic Orbifold Enumeration}
\label{sec:computation}

In Table~\ref{tbl:polygon_orbifold} we provide an enumeration of 2D kaleidoscopic orbifolds based on the cardinality of their underlying polygons and the type of their universal covers. To justify this enumeration, we organize the computation behind this enumeration into a number of theorems in this section.

\begin{theorem}
  Given a kaleidoscopic orbifold $O=$\textbf{$*k_1k_2\dots k_N$}, $O$ is a spherical orbifold if $N\le 2$ and a hyperbolic orbifold if $N>4$. When $N=4$, $O$ is a hyperbolic orbifold with the only exception of \textbf{$*2222$}, which is a Euclidean orbifold.
  \label{thm:theorem_n_walls}
\end{theorem}

\begin{proof}

To show these statements, we only need to compute the Euler characteristics of these orbifolds using Equation~\ref{eq:euler} in the paper which states that Euler characteristic of a kaleidoscopic orbifold $O=$\textbf{$*k_1k_2\dots k_N$} is $\chi(O) = \sum_{i=1}^{N} \frac{1}{2k_i} -\frac{N}{2}+1$.

When $N=1$, the polygon is a monogon, with one mirror wall that self-intersects at an angle of $\frac{\pi}{k}$. The only good kaleidoscopic orbifold is when $k=1$, i.e. the wall self-intersects at an angle of $\pi$. Note that the orbifold $O=$\textbf{$*1$} is usually abbreviated as \textbf{$*$}. Thus, $\chi(O)=\frac{1}{2} - \frac{1}{2}+1=1>0$. Consequently, this orbifold is spherical. In fact, this orbifold is a hemisphere with the boundary having the reflectional symmetry.

When $N=2$, $O=$\textbf{$*k_1k_2$} where $1<k_1\le k_2$. Consequently, $\chi(O)=\frac{1}{2k_1}+\frac{1}{2k_2} - \frac{2}{2}+1=\frac{1}{2k_1}+\frac{1}{2k_2}>0$, which means that this type of orbifolds are also spherical.

On the other hand, when $N \ge 4$, we have $k_i \ge 2$ for any $1\le i \le N$. Thus, $\chi(O)=\sum_{i=1}^{N}\frac{1}{2k_i} - \frac{N}{2}+1 \le \frac{1}{4} N - \frac{N}{2}+1 =1-\frac{N}{4} \le 0$. Notice that the equality holds in the above only when $N=4$ and $k_1=k_2=k_3=k_4=2$. Consequently, this type of orbifolds is hyperbolic with the only exception of \textbf{$*2222$}, which is a Euclidean orbifold.

\end{proof}

Theorem~\ref{thm:theorem_n_walls} indicates that the more walls there are in the kaleidoscopic orbifold, the more negative its Euler characteristic and the more likely the orbifold being hyperbolic. In contrast, the fewer the walls the more positive its Euler characteristic and more likely the orbifold being spherical. The boundary between the set of spherical orbifolds and the set of hyperbolic orbifolds is drawn when $N=3$, i.e. triangular orbifolds. The next several theorem inspect this scenario, which consists a number of cases.

\begin{theorem}
  Given a triangular kaleidoscopic orbifold $O=$\textbf{$*k_1k_2k_3$} and without the loss of generality assuming that $3\le k_1 \le k_2 \le k_3$, $O$ is a hyperbolic orbifold with the only exception of \textbf{$*333$}, which is a Euclidean orbifold.
  \label{thm:theorem_triangular_min}
\end{theorem}

\begin{proof}

First of all, the assumption that $3\le k_1 \le k_2 \le k_3$ makes sense since any permutation of $k_1$, $k_2$, and $k_3$ gives rise the same triangular orbifold.

Again, we only need to compute the Euler characteristics of these orbifolds. Here, $\chi(O)=\sum_{i=1}^{3}\frac{1}{2k_i} - \frac{3}{2}+1 \le \frac{1}{6} 3 - \frac{3}{2}+1 = 0$. Notice that the equality holds in the above only when $k_1=k_2=k_3=3$. Consequently, this type of orbifolds is hyperbolic with the only exception of \textbf{$*333$}, which is a Euclidean orbifold.

\end{proof}

Theorem~\ref{thm:theorem_triangular_min} states that for triangular orbifolds, the higher the minimal order of symmetry at the corners, namely $k_1$, the more likely the orbifold is hyperbolic. We now consider the case when $k_1=2$.

\begin{theorem}
  Given a triangular kaleidoscopic orbifold $O=$\textbf{$*2k_2k_3$} where $2 \le k_2 \le k_3$, $O$ is a spherical orbifold if $k_2=2$. In contrast, when $k_2\ge 4$, $O$ is a hyperbolic orbifold with the only exception of \textbf{$*244$}, which is a Euclidean orbifold.
  \label{thm:theorem_triangular_2_min}
\end{theorem}

\begin{proof}

Since $N=3$, when $k_1=k_2=2$ we find the Euler characteristic $\chi(O)=\frac{1}{4} + \frac{1}{4}+\frac{1}{2k_3} - \frac{3}{2}+1 =\frac{1}{2k_3}>0$. Thus, in this case the orbifold is always spherical.

On the other hand, when $k_3\ge k_2\ge 4$, the Euler characteristics is $\chi(O)=\frac{1}{4} + \frac{1}{2k_2}+\frac{1}{2k_3}- \frac{3}{2}+1 \le \frac{1}{4} + \frac{1}{8}+\frac{1}{8}- \frac{3}{2}+1 = 0$. Notice that the equality holds in the above only when $k_2=k_3=4$. Consequently, this type of orbifolds is hyperbolic with the only exception of \textbf{$*244$}, which is a Euclidean orbifold.

\end{proof}

The last remaining case is when $O=$\textbf{$*23k_3$}, which is covered in the next theorem.

\begin{theorem}
  Given a triangular kaleidoscopic orbifold $O=$\textbf{$*23k_3$} where $3\le k_3$, $O$ is a spherical orbifold if $k_3<6$, a Euclidean orbifold if $k_3=6$, and a hyperbolic orbifold if $k_3>6$.
  \label{thm:theorem_triangular_23k}
\end{theorem}

\begin{proof}

The Euler characteristic of this type of orbifolds is $\chi(O)=\frac{1}{4} + \frac{1}{6}+\frac{1}{2k_3} - \frac{3}{2}+1 = \frac{1}{2k_3}-\frac{1}{12}=\frac{6-k_3}{12k_3}$. Thus, $\chi(O)$ is positive when $k_3<6$, zero when $k_3=6$, and negative when $k_3>6$. Consequently, $O$ is a spherical orbifold if $k_3<6$, a Euclidean orbifold if $k_3=6$, and a hyperbolic orbifold if $k_3>6$.

\end{proof}

Interestingly, each of the above theorems contains a Euclidean orbifold: \textbf{$*2222$} for Theorem~\ref{thm:theorem_n_walls}, \textbf{$*333$} for Theorem~\ref{thm:theorem_triangular_min}, \textbf{$*244$} for Theorem~\ref{thm:theorem_triangular_2_min}, and \textbf{$*236$} for Theorem~\ref{thm:theorem_triangular_23k}. Not only do these facts confirm that there are only four Euclidean kaleidoscopic orbifolds, but they also show the transition from spherical orbifolds to hyperbolic orbifolds with more walls and higher-order symmetries at the corners.

\begin{figure*}[t]
	\centering%
	$\begin{array}{@{\hspace{0.0in}}c@{\hspace{0.1in}}c@{\hspace{0.1in}}c}
	\includegraphics[width=2.3in]{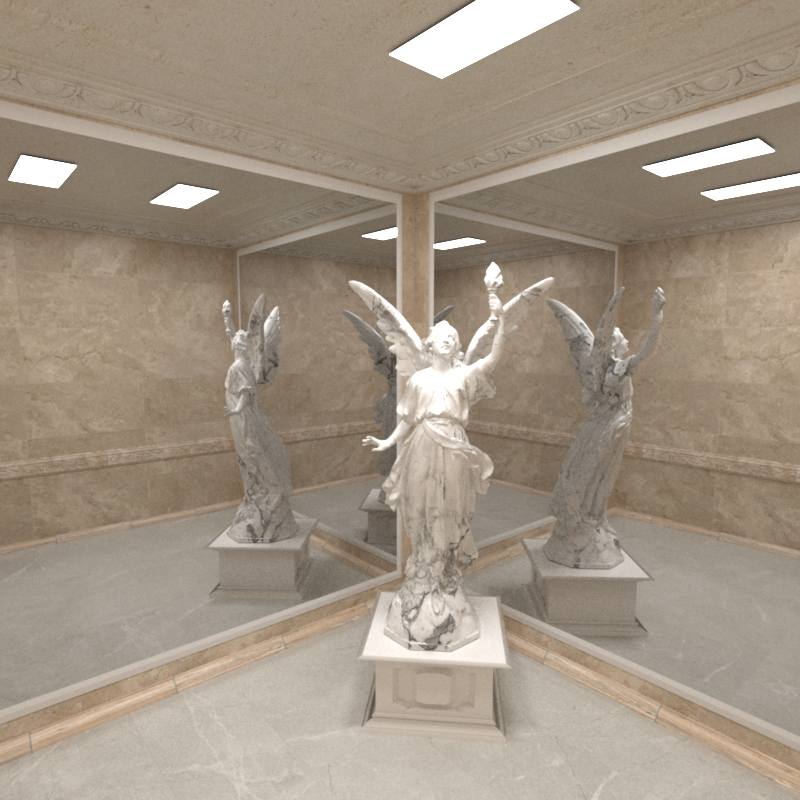}
	&\includegraphics[width=2.3in]{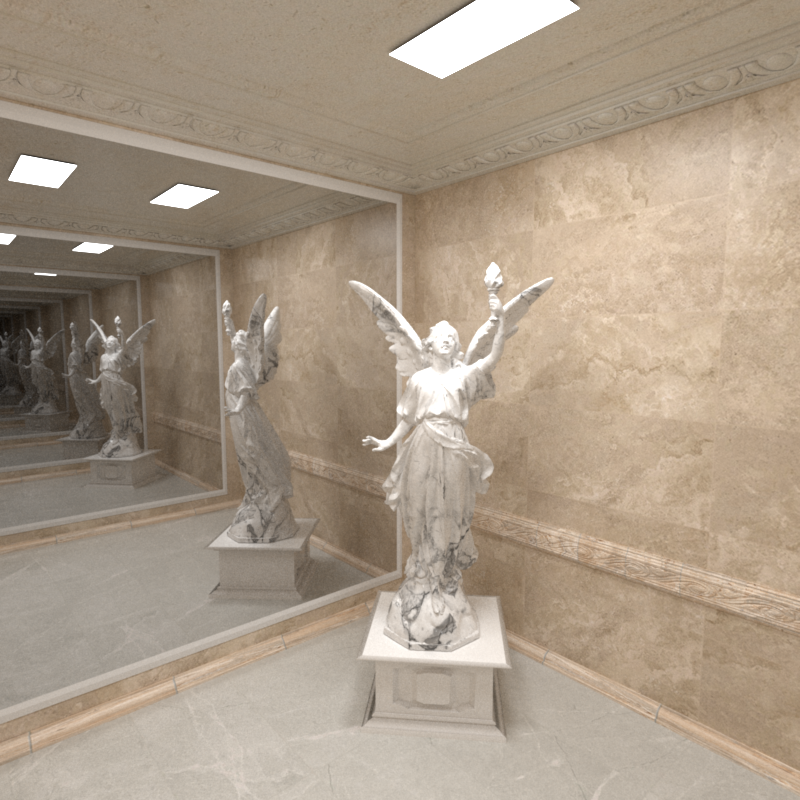}
	&\includegraphics[width=2.3in]{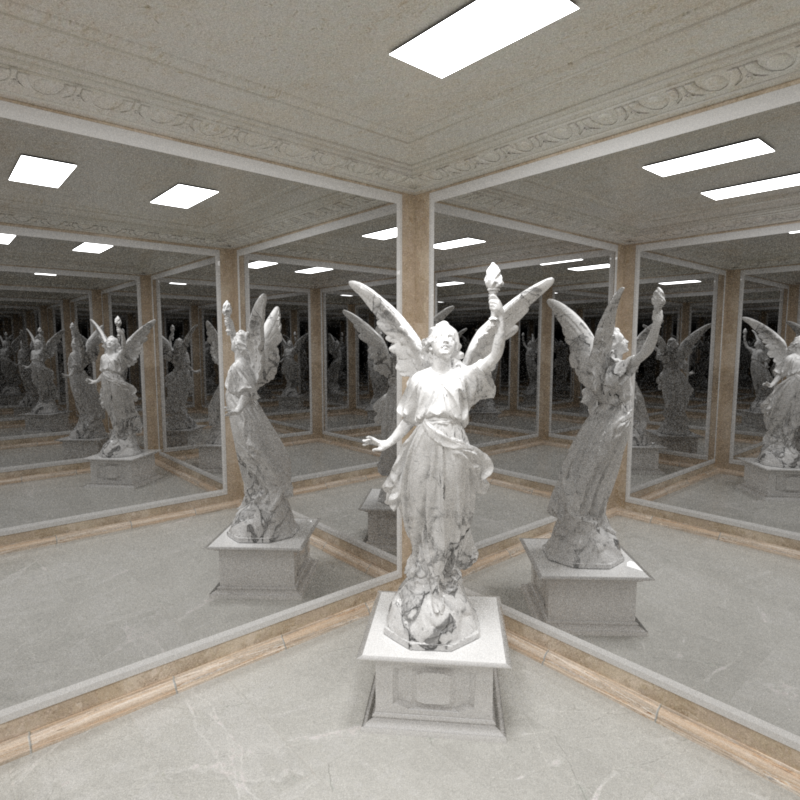}  \\
	(a) \quad \textbf{two adjacent mirrors} & (b) \quad \textbf{two parallel mirrors} & (c) \quad \textbf{four mirrors} \\
	\end{array}$
	\caption[]{A square room with two or four mirrors. The case in the four mirror room (c) corresponds to a Euclidean orbifold \textbf{$*2222$}. }
	\label{fig:square_room}
\end{figure*}

\begin{figure*}[t]
	\centering%
	$\begin{array}{@{\hspace{0.0in}}c@{\hspace{0.1in}}c@{\hspace{0.1in}}c}
	\includegraphics[width=2.3in]{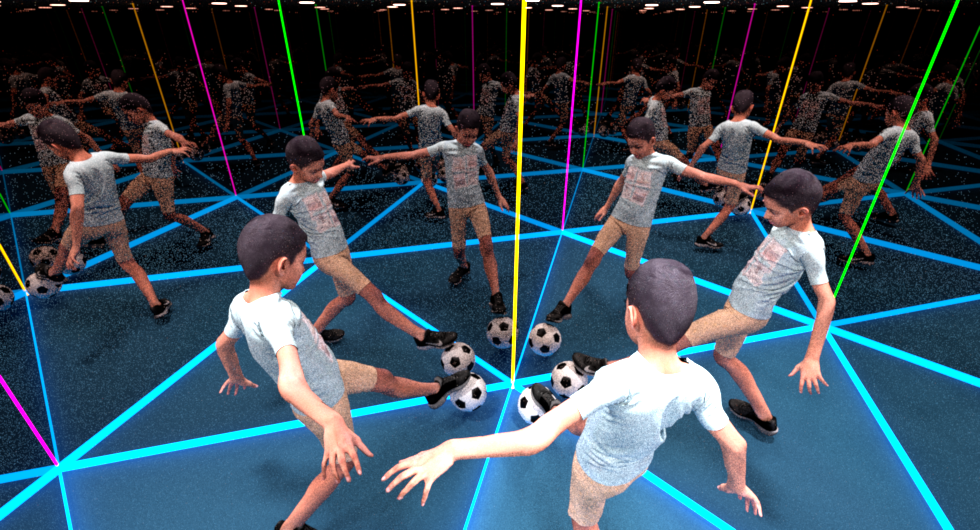}
	&\includegraphics[width=2.3in]{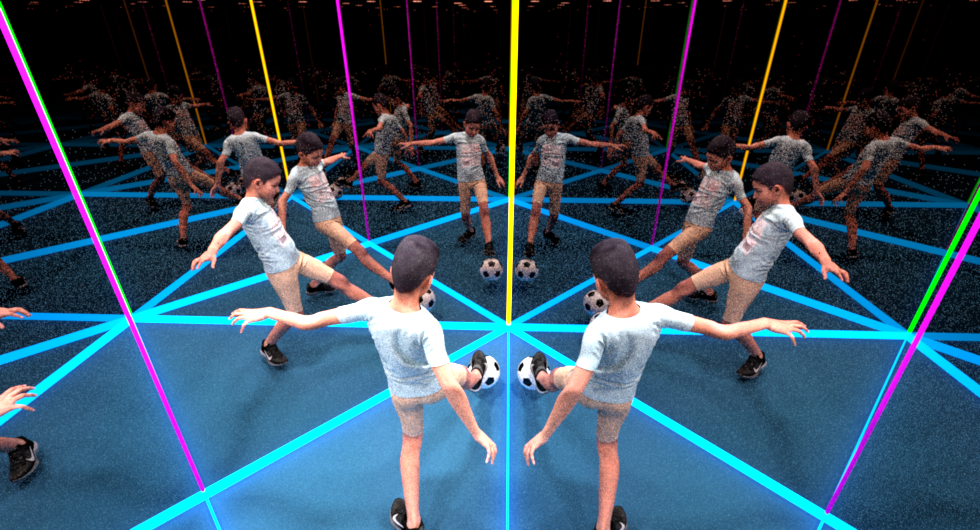}
	&\includegraphics[width=2.3in]{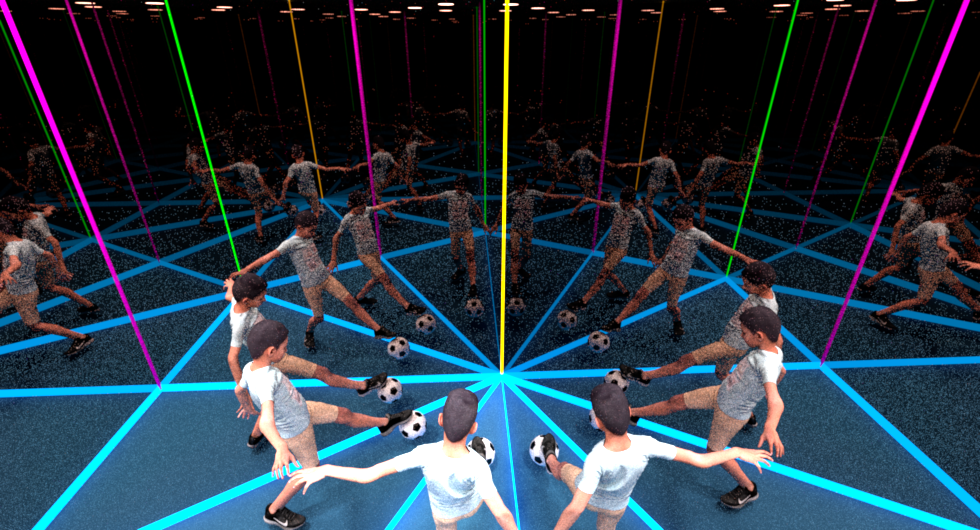}  \\
	(a) \quad \textbf{*$333$} & (b) \quad \textbf{$*244$} & (c) \quad \textbf{$*236$} \\
	\end{array}$
	\caption[]{The three triangular Euclidean kaleidoscopic orbifolds.   }
	\label{fig:triangular_room}
\end{figure*}

\section{Optics-Based Visualization for Orbifold Concept and Properties}
\label{sec:more_visualization_results}

Our system can be used to generate example scenarios to illustrate important concepts and properties of orbifolds such as the following. Given a room with the statue Lucy, we first mount a mirror each on two adjacent walls (Figure~\ref{fig:square_room} (a)). This leads to an illusion of a space that is four times as large as the room without a mirror. The virtual space is the {\em universal cover} of the orbifold (the original room).

In addition, the symmetry for the room can be understood by checking the orientations of the statues in the space. While the statue has her left hand up holding the torch in the original room, each mirror generates a virtual statue who raises the torch by her right hand (a reflection). Interestingly, reflecting the statue in the first virtual room with respect to the second mirror leads to the third virtual statue, who switches back to her left hand to raise the torch. However, this virtual statue faces the opposite direction of the statue in the original room, i.e. a rotation by $\pi$. One can consider the reflected and rotated virtual copies as the result of the {\em action} of the {\em symmetry group} of the underlying orbifold. This group consists of the identity action, two reflections (one per each mirror), and one rotation (the composite of the two mirrors).

By moving one of the mirrors to the wall opposite the other mirror, we obtain a different scene where there are infinitely many copies of the original room (Figure~\ref{fig:square_room} (b)). In fact, the universal cover of this orbifold can be generated by first grouping the original room with one of the reflections and then translating infinitely many times the two rooms by a distance that is a multiple of twice the room depth. The union of the two rooms (the real room and the virtual room) is thus referred to as a {\em translational cover}.

When a mirror is mounted on each wall (Figure~\ref{fig:square_room} (c)), we obtain the orbifold whose translational cover is the same as the universal cover of the room shown in Figure~\ref{fig:square_room} (a). This translational cover is then translated in two mutually perpendicular directions. Note that this is the first orbifold (in this example) that we have encountered where all walls have a mirror. This room corresponds to the \textbf{$*2222$} orbifold. Each of the corner has an angle of $\frac{\pi}{2}$, thus its notation. At such a corner, there are $2k$ copies of the original room forming $k$ pairs. Inside each pair, one of the rooms is a rotational copy of the original room while the other is a reflectional copy. A {\em kaleidoscopic orbifold} has a {\em transformation group} that is generated by mirror reflections. The {\em subgroup} for each corner is thus $\mathbb{D}_k$, the {\em Dihedral group} of {\em order} $k$. In the \textbf{$*2222$} case, the symmetry group at every corner is the same, i.e. $\mathbb{D}_2$.

\textbf{$*2222$} is one of the four Euclidean kaleidoscopic orbifolds, i.e. whose universal cover is the Euclidean plane. Figure~\ref{fig:triangular_room} shows the other three such orbifolds: (1) \textbf{$*333$}, (2) \textbf{$*244$}, and (3) \textbf{$*236$}. The \textbf{$*333$} orbifold (Figure~\ref{fig:triangular_room} (a)) is obtained by placing three mirror walls in a $\frac{\pi}{3}-\frac{\pi}{3}-\frac{\pi}{3}$ triangular room. Its translational cover consists of six copies of the original room ($\mathbb{D}_3$). Similarly, the \textbf{$*244$} orbifold (Figure~\ref{fig:triangular_room} (b)) is obtained by placing three mirror walls in a $\frac{\pi}{2}-\frac{\pi}{4}-\frac{\pi}{4}$ triangular room. Its translational cover consists of eight copies of the original room ($\mathbb{D}_4$). The \textbf{$*236$} orbifold (Figure~\ref{fig:triangular_room} (c)) is generated by placing three mirror walls in a $\frac{\pi}{2}-\frac{\pi}{3}-\frac{\pi}{6}$ triangular room. Its translational cover consists of $12$ copies of the original room ($\mathbb{D}_6$). Notice that the symmetry group can vary from corner to corner. In addition, note that an orbifold does not depend on which corner is referred to as the first corner. Thus, \textbf{$*236$} and \textbf{$*362$} represent the same orbifold. Similarly, the orbifold does not change when the corners are numbered in the opposite order. Thus, \textbf{$*236$} and \textbf{$*632$} also represent the same orbifold.

The optics-based visual metaphor is also capable of showing non-Euclidean orbifolds, such as those shown in Figures~\ref{fig:teaser},~\ref{fig:str_curved_rays},
and~\ref{fig:non_Euclidean_room} in the main paper.

\end{document}

%% file: setting.tex
\long\def\symbolfootnote[#1]#2{\begingroup%
\def\thefootnote{\fnsymbol{footnote}}\footnote[#1]{#2}\endgroup}

\newtheorem{theorem}{\sffamily Theorem}